\documentclass[a4paper,UKenglish, cleveref, thm-restate, numberwithinsect]{lipics-v2021}
\usepackage{microtype}
\usepackage{amssymb} %
\usepackage{algorithm}
\usepackage{algpseudocode}

\usepackage{xcolor,pgf,tikz,pgflibraryarrows,pgffor,pgflibrarysnakes}
\usepackage{mathtools}
\usepackage{complexity}
\usepackage{xparse}

\makeatletter
\newcommand*{\ol}[1]{$\overline{\hbox{#1}}\m@th$}
\makeatother

\usetikzlibrary{fit} %
\usetikzlibrary{backgrounds} %

\usepackage{blkarray}%
\usepackage{slashbox}
\usepgflibrary{shapes}
\usetikzlibrary{snakes,automata}

\tikzstyle{background}=[rectangle,fill=gray!10, inner sep=0.1cm, rounded corners=0mm]

\usepackage{tikz}
\usepackage{color,colortbl}
\usetikzlibrary{automata,positioning}

\tikzstyle{cir}=[fill=green,circle,minimum size=.4em,inner sep=0em] 
\tikzstyle{cirr}=[draw=black, fill=black!80,circle,minimum size=.4em,inner sep=0em]

\tikzstyle{cir1}=[fill=red,circle,minimum size=.4em,inner sep=0em]                                                                                                      

\tikzstyle{cir2}=[fill=violet,circle,minimum size=.4em,inner sep=0em]                                                                                                      

\DeclareDocumentCommand{\BSREACH}{}{
   \mathsf{BSREACH}
}

\newcommand{\cG}{\mathcal{G}}
\newcommand{\cA}{\mathcal{A}}
\newcommand{\calP}{\mathcal{P}}

\newcommand{\Nn}{\mathbb{N}}
\newcommand{\Zn}{\mathbb{Z}}

\newcommand{\mini}{{\mathop{\mathsf{min}}\,}}
\newcommand{\maxi}{{\mathop{\mathsf{max}}\,}}

\usepackage{todonotes}
\newcounter{todocounter}

\usepackage{amssymb}
\usepackage{stmaryrd}

\newcommand\assigned\leftarrow

\newcommand\lbl\lambda

\newcommand{\Aa}{\mathcal{A}}

\newcommand{\bbM}{\mathbb{M}}

\usepackage{listings}
\usepackage{multirow}

\newcommand{\autstep}{\rightarrow}
\newcommand{\autsteps}{\xrightarrow{*}}

\usepackage{amsmath}
\makeatletter
\newcommand\sr[2][]{\ext@arrow 0099{\longrightarrowfill@}{#1}{#2}}
\def\longrightarrowfill@{\arrowfill@\relbar\rightarrow}
\makeatother

\newcommand{\N}{\mathbb{N}}

\newcommand{\xsubparagraph}[1]{
  \vspace{0.1cm}
  \noindent
  {{\sffamily\normalsize\bfseries #1}}}

\title{Scope-Bounded Reachability in Valence Systems} %

\titlerunning{Scope-Bounded Valence Systems} %

\author{Aneesh K. Shetty}{Department of Computer Science \& Engineering, IIT Bombay}{}{}{}

\author{S. Krishna}{Department of Computer Science \& Engineering, IIT Bombay}{}{https://orcid.org/0000-0003-0925-398X}{}

\author{Georg Zetzsche}{Max Planck Institute for Software Systems, Kaiserslautern, Germany}{}{https://orcid.org/0000-0002-6421-4388}{}

\hideLIPIcs
\nolinenumbers

\authorrunning{K. Aneesh, S. Krishna, G. Zetzsche} %

\Copyright{Aneesh K. Shetty, S. Krishna, Georg Zetzsche} %



\category{} %

\relatedversion{} %

\supplement{}%

\begin{document}

\maketitle

\begin{abstract}
  Multi-pushdown systems are a standard model for concurrent recursive programs, but they have an undecidable reachability problem.
  Therefore, there have been several proposals to underapproximate their sets of runs so that reachability in this underapproximation becomes decidable.
  One such underapproximation that covers a relatively high portion of runs is \emph{scope boundedness}.
  In such a run, after each push to stack $i$, the corresponding pop operation must come within a bounded number of visits to stack $i$.

  In this work, we generalize this approach to a large class of infinite-state systems.
  For this, we consider the model of valence systems, which consist of a finite-state control and an infinite-state storage mechanism that is specified by a finite undirected graph.
  This framework captures pushdowns, vector addition systems, integer vector addition systems, and combinations thereof.
  For this framework, we propose a notion of scope boundedness that coincides with the classical notion when the storage mechanism happens to be a multi-pushdown.

  We show that with this notion, reachability can be decided in $\PSPACE$ for every storage mechanism in the framework.
  Moreover, we describe the full complexity landscape of this problem across all storage mechanisms, both in the case of (i)~the scope bound being given as input and (ii)~for fixed scope bounds.
  Finally, we provide an almost complete description of the complexity landscape if even a description of the storage mechanism is part of the input.
\end{abstract}

\section{Introduction}\label{sec:introduction}

Multi-pushdown systems are a natural model for recursive programs with threads that communicate via shared memory.
Unfortunately, even safety verification (state reachability) is undecidable for this model~\cite{ramalingam2000context}.
However, by considering \emph{underapproximations} of the set of all executions, it is still possible to discover safety violations.
The first such underapproximation in the literature was \emph{bounded context switching}~\cite{qadeer2005context}.
Here, one only considers executions that switch between threads a bounded number of times.
In terms of multi-pushdown systems, this places a bound on the number of times we can switch between stacks.

One underapproximation that covers a relatively large portion of all executions and still permits decidable reachability is \emph{scope-boundedness} as proposed by La Torre, Napoli, and Parlato~\cite{tcs12,ic20}.
Here, instead of bounding the number of context switches across the entire run, we bound the number of context switches \emph{per letter on a stack} (i.e. procedure execution).
More precisely, whenever we push a letter on some stack $i$, then we can switch back to stack $i$ at most $k$ times before we have to pop that letter again.
This higher coverage of runs comes at the cost of higher complexity: While reachability with bounded context switching is $\NP$-complete~\cite{DBLP:journals/toplas/EsparzaGP14,qadeer2005context}, the scope-bounded reachability problem is $\PSPACE$-complete (if the number of pushdowns or the scope bound is part of the input)~\cite{ic20}.

Aside from multi-pushdown systems, there is a wide variety of other infinite-state models that are used to model program behaviours.
For these, reachability problems are also sometimes undecidable or have prohibitively high complexity. 
For example, vector addition systems with states (VASS) is one of the most prominent models for concurrent systems, but its reachability problem has non-elementary complexity~\cite{DBLP:conf/stoc/CzerwinskiLLLM19}.
This raises the question of whether underapproximations for multi-pushdown systems can be interpreted in other infinite-state systems and what complexity would ensue.

The notion of bounded context switching has recently been generalized to a large class of infinite-state systems~\cite{DBLP:conf/concur/MeyerMZ18}, in the framework of \emph{valence systems over graph monoids}. 
These consist of a finite-state control that has access to a storage mechanism.
The shape of this storage mechanism is described by a finite, undirected graph.
By choosing an appropriate graph, one can realize many infinite-state models.
Examples include (multi-)pushdown systems, VASS, integer VASS, but also combinations thereof, such as pushdown VASS~\cite{DBLP:conf/icalp/LerouxST15} and sequential recursive Petri nets~\cite{haddad2007recursive}. 
Under this notion, bounded context reachability is in $\NP$ for each graph, and thus each storage mechanism in the framework~\cite{DBLP:conf/concur/MeyerMZ18}. 
Moreover, the paper \cite{DBLP:conf/concur/MeyerMZ18} presents some subclasses of graphs for which bounded context reachability has lower complexity ($\NL$ or $\P$). However, the exact complexity of reachability under bounded context switching remains open in many cases, such as the path with four nodes~\cite{DBLP:conf/concur/MeyerMZ18}.

\xsubparagraph{Contribution}
We present an \emph{abstract notion of scope-bounded runs} for valence systems over graph monoids. 
As we show, this notion always leads to a reachability problem decidable in $\PSPACE$. 
In particular, our notion applies to all infinite-state models mentioned above. 
Moreover, applied to multi-pushdown systems, it coincides with the notion of La Torre, Napoli, and Parlato.

We also obtain an almost complete complexity landscape of scope-bounded reachability.
First, we show that if both (i)~the graph $\Gamma$ describing the storage mechanism and (ii)~the scope bound $k$ are part of the input, the problem is $\PSPACE$-complete.
Second, we study how the complexity depends on the employed storage mechanism.
We show that for each $\Gamma$, the problem is either $\NL$-complete, $\P$-complete, or $\PSPACE$-complete, depending on $\Gamma$ (\cref{main:k-input-graph-fixed}).
Since the complexity drops below $\PSPACE$ only in extremely restricted cases, we also study the setting where the scope bound $k$ is fixed.
In this case, we show that the problem is either $\NL$-complete or $\P$-complete, depending on $\Gamma$ (\cref{main:k-fixed-graph-fixed}). 

Finally, applying scope-boundedness to classes of infinite-state systems requires understanding the complexity if $\Gamma$ is drawn from an infinite class of graphs.
For example, for each fixed dimension $d$, there is a graph $\Gamma_d$ such that valence automata over $\Gamma_d$ correspond to VASS of dimension $d$.
The class of \emph{all} VASS (of arbitrary dimension), however, corresponds to valence automata over all cliques.
Thus, we also study scope-bounded reachability if $\Gamma$ is restricted to a class of graphs $\cG$.
Under a mild assumption on $\cG$, we again obtain a complexity trichotomy of $\NL$-, $\P$-, or $\PSPACE$-completeness, both for $k$ as input (\cref{main:k-input}) and for fixed $k$ (\cref{main:k-fixed}). In fact, all results mentioned above follow from these general results.

\xsubparagraph{Related work}
Similar in spirit to our work are the lines of research on systems with bounded
tree-width by Madhusudan and Parlato~\cite{DBLP:conf/popl/MadhusudanP11} and on
bounded split-width by Aiswarya~\cite{DBLP:phd/hal/Cyriac14}. In these
settings, the storage mechanism is represented as a class of possible matching
relations on the positions of a computation. Then, under the assumption that
the resulting \emph{behavior graphs} have bounded tree-width or split-width,
respectively, there are general decidability results. In particular,
decidability of scope-bounded reachability in multi-pushdown systems 
has been deduced via tree-width~\cite{DBLP:conf/fsttcs/TorreP12} and via
split-width~\cite{DBLP:conf/concur/CyriacGK12}. Different from
underapproximations based on bounded tree-width or split-width, our framework
includes multi-counter systems (such as VASS or integer VASS), but also
counters nested in stacks. While VASS can be seen as special cases of
multi-pushdown systems, our framework allows us, e.g. to study the complexity
of scope-bounded reachability if the storage mechanism is restricted to
multi-counters. On the other hand, while tree-width and split-width can be
considered for queues~\cite{DBLP:conf/popl/MadhusudanP11,DBLP:conf/atva/AiswaryaGK14}, they
cannot be realized as storage mechanisms in valence systems.

Furthermore, after their introduction~\cite{tcs12} scope-bounded multi-pushdown
systems have been studied in terms of accepted languages~\cite{DBLP:journals/ijfcs/TorreNP16}, temporal logic
model checking~\cite{DBLP:conf/fsttcs/TorreP12,DBLP:conf/atva/AtigBKS12}.
Moreover, scope-boundedness has been studied in the timed
setting~\cite{DBLP:conf/tacas/AkshayGKR20},\cite{DBLP:conf/dlt/BhaveKPT19}.  

Over the last decade, the framework of valence automata over graph monoids has been used to study how several types of analysis are impacted by the choice of storage mechanism.
For example: For which storage mechanisms (i)~can silent transitions be algorithmically eliminated?~\cite{DBLP:conf/icalp/Zetzsche13}; (ii)~do we have a Parikh's theorem~\cite{DBLP:conf/mfcs/BuckheisterZ13}, (iii)~is (general) reachability decidable~\cite{DBLP:journals/iandc/Zetzsche21}; (iv)~is first-order logic with reachability decidable?~\cite{DBLP:conf/lics/DOsualdoMZ16}; (v)~can downward closures be computed effectively?~\cite{DBLP:conf/stacs/Zetzsche15}.

Details of all proofs can be found in the full version of the paper.

\section{Preliminaries}\label{sec:preliminaries}
In this section, we recall the basics of valence systems over graph monoids \cite{georgthesis}.

\xsubparagraph{Graph Monoids}
This class of monoids accommodate a variety of storage mechanisms.
They are defined by undirected graphs without parallel edges $\Gamma=(V, I)$ where $V$ is a finite set of vertices and $I \subseteq \{e\subseteq V \mid 1\le |e|\le 2\}$ is a finite set of undirected edges, which can be self-loops.
Thus, if $\{v\}\in I$, we say that $v$ is \emph{looped}; otherwise, $v$ is \emph{unlooped}.
The edge relation is also called an \emph{independence relation}.
We also write $uIv$ for $\{u,v\}\in I$.
A subset $U\subseteq V$ is a \emph{clique} if $uIv$ for any two distinct $u,v\in U$.
If in addition, all $v\in U$ are looped, then $U$ is a \emph{looped clique}.
If $U$ is a clique and all $v\in U$ are unlooped, then $U$ is an \emph{unlooped clique}.
We say that $U\subseteq V$ is an \emph{anti-clique} if we do not have $uIv$ for any distinct $u,v\in U$.
Given the graph, we define a monoid as follows.
We have the alphabet $X_\Gamma=\{v^+,v^-\mid v\in V\}$, where we write $xIy$ for $x,y\in X_\Gamma$ if for some $u,v\in V$, we have $x\in\{u^+,u^-\}$, $y\in\{v^+,v^-\}$, $x\ne y$, and $uIv$. Moreover, $\equiv_\Gamma$ is the smallest congruence on $X_\Gamma^*$ with $v^+v^-\equiv_\Gamma \varepsilon$ for $v\in V$ and $xy\equiv_\Gamma yx$ for $xIy$. Here, $\varepsilon$ denotes the empty word.
Thus, if $v$ has a self-loop, then $v^-v^+\equiv_\Gamma\varepsilon$.
We define the monoid $\bbM\Gamma:=X_{\Gamma}^*/\equiv_\Gamma$.

\xsubparagraph{Valence Systems}
Graph monoids are used in valence systems, which are finite automata whose edges are labeled with elements of a monoid.
Then, a run is considered valid if the product of the monoid elements is the neutral element.
Here, we only consider the case where the monoid is of the form $\bbM\Gamma$, so we define the concept directly for graphs.

Given a graph $\Gamma$, a valence system $\Aa$ over $\Gamma$ consists of a finite set of states $Q$, and a transition relation $\rightarrow\subseteq Q \times X_{\Gamma}^* \times Q$.
A \emph{configuration} of $\Aa$ is a tuple $(q, w)$ where $q \in Q$, $w \in X_{\Gamma}^*$ is the sequence of storage operations executed so far.
From a configuration $(q_1, u)$, on a transition $q_1 \stackrel{v}{\rightarrow} q_2$, we reach the configuration $(q_2,uv)$.
A run of $\Aa$ is a sequence of transitions.
The \emph{reachability problem} in valence systems is the following: Given states $q_{\mathit{init}}$ and $q_{\mathit{fin}}$, is there a run from $(q_{\mathit{init}}, \varepsilon)$ that reaches $(q_{\mathit{fin}}, w)$ for some $w\in X_\Gamma^*$ with $w\equiv_\Gamma \varepsilon$?

Many classical storage types can be realized with graph monoids.
Consider $\bar{\Gamma}_3 = (V, I)$ in \cref{valence}.
We have $I = \{\{a, c\}, \{b,c\}, \{c\}\}$.
For $w\in X_{\bar{\Gamma}_3}^*$ we have $w\equiv_{\bar{\Gamma}_3}\varepsilon$ if and only if two conditions are met: First, if we project to $\{a^+,a^-,b^+,b^-\}$, then the word corresponds to a sequence of push- and pop-operations that transform the empty stack into the empty stack.
Here, $x^+$ corresponds to pushing $x$, and $x^-$ to popping $x$, for $x\in\{a,b\}$.
Second, the number of $c^+$ is the same as the number of $c^-$ in $w$.
Thus, valence automata over $\bar{\Gamma}_3$ can be seen as pushdown automata that have access to a $\Zn$-valued counter.
Similarly, the storage mechanism of $\Gamma_2$ in \cref{valence} is a stack, where each stack entry is not a letter, but contains two $\Nn$-valued counters.
A push ($c^+$) starts a new stack entry and a pop ($c^-$) is only possible if the topmost two counters are zero.
For more examples and explanation, see~\cite{DBLP:journals/eatcs/Zetzsche16}.

\newcommand{\pushdown}[1]{\mathsf{P}_{#1}}
\newcommand{\multipushdown}[2]{\mathsf{MP}_{#1,#2}}
\newcommand{\uclique}[1]{\mathsf{UC}_{#1}}
\newcommand{\lclique}[1]{\mathsf{LC}_{#1}}
\newcommand{\cliquem}[1]{\mathsf{UC}^-_{#1}}
\newcommand{\stackedcounter}[1]{\mathsf{SC}_{#1}}

\begin{example}[Example storage mechanisms]\label{storage-examples}
	Let us mention a few particular (classes of) graphs and how they correspond to infinite-state systems.
	In the following, the \emph{direct product} of two graphs $\Gamma$ and $\Delta$ is the graph obtained by taking the disjoint union of $\Gamma$ and $\Delta$ and adding an edge between each vertex from $\Gamma$ and each vertex from $\Delta$.

	\begin{description}
		\item[Pushdown] For $s\in\N$, let $\pushdown{s}$ be the graph on $s$ vertices without edges.
			Then valence automata over $\pushdown{s}$ correspond to pushdown systems with $s$ stack symbols.

		\item[Multi-pushdown] Let $\multipushdown{r}{s}$ be the direct product of $r$ disjoint copies of $\pushdown{s}$.
			Then valence systems over $\multipushdown{r}{s}$ correspond to multi-pushdown systems with $r$ stacks, each of which has $s$ stack symbols.
			In Figure \ref{valence}, the induced subgraph of graph $\Gamma_1$ on $\{b_1, b_2, b_3, c_1, c_2, c_3\}$ represents $\multipushdown{2}{3}$.

		\item[VASS] If $\uclique{d}$ is an unlooped clique with $d$ vertices, then valence systems over $\uclique{d}$ correspond to $d$-dimensional vector addition systems with states.

		\item[Integer VASS] If $\lclique{d}$ be a looped clique with $d$ vertices, then valence systems over $\lclique{d}$ correspond to $d$-dimensional integer VASS.

		\item[Pushdown VASS] If $\cliquem{d}$ is the graph obtained from $\uclique{d+2}$ by removing a single edge, then valence systems over $\cliquem{d}$ correspond to $d$-dimensional pushdown VASS.
	\end{description}

\end{example}

\section{Scope-bounded runs in valence systems}\label{sec:scope-bounded}
In this section, we introduce our notion of bounded scope to valence systems over arbitrary graph monoids. For each of the used concepts, we will explain how they relate to the existing notion of scope-boundedness for multi-pushdown systems.
Fixing $\Gamma=(V,I)$ as before, first we introduce some preliminary notations and definitions.

\xsubparagraph{Dependent sets and contexts}
Recall that valence systems over the graph $\multipushdown{r}{s}$ realize a
storage consisting of $r$ pushdowns, each with $s$ stack symbols. The graph
$\multipushdown{r}{s}$ is a direct product of $r$-many disjoint anti-cliques,
each with $s$ vertices. Here, each anti-clique corresponds to a pushdown with
$s$ stack symbols: For a vertex $v$ in such an anti-clique, the symbol $v^+$ is
the push operation for this stack symbol, and $v^-$ is its pop operation.

In a multi-pushdown system, a run is naturally decomposed into contexts, where
each context is a sequence of operations belonging to one stack.
In~\cite{DBLP:conf/concur/MeyerMZ18}, the notion of context was generalized to
valence systems as follows.  A set $U\subseteq V$ is called \emph{dependent} if
it does not contain distinct vertices $u_1, u_2 \in V$ such that $u_1 I u_2$.
A set of operations $Y\subseteq X_\Gamma$ is \emph{dependent} if its underlying
set of vertices $\{v \in  V \mid v^+\in Y$ or $v^-\in Y\}$ is dependent.  A
computation is called \emph{dependent} if the set of operations occurring in it
is dependent.   A dependent computation is also called a \emph{context}.  In
$\Gamma_1$ of Figure \ref{valence}, contexts can be formed over $\{b_1, b_2,
b_3\}$, $\{c_1, c_2, c_3\}$ and $\{a\}$. 

\smallskip

\xsubparagraph{Context decomposition}
Note that a word $w\in X_\Gamma^*$ need not have a unique decomposition into contexts.
For example, for $\Gamma_2$ in \cref{valence}, the word $a^+c^+b^+$ can be decomposed as $(a^+c^+)b^+$ and as $a^+(c^+b^+)$. Therefore, we now define a canonical decomposition into
contexts, which decomposes the word from left to right.
Formally, the \emph{canonical context decomposition} of a computation $w \in X_\Gamma^+$ (that is, $|w| > 0$) is defined inductively. 
If $w$ is over a dependent set of operations, then $w$ is a single context. 
Otherwise, find the maximal, non-empty prefix $w_1$ of $w$ over a dependent set of operations. The canonical decomposition of $w$ into contexts is then $w = w_1 w_2 \dots w_m$ where $w_2 \dots w_m$ is the decomposition of the remaining word into contexts. 
In the following, when we mention the contexts of a word, we always mean those in the canonical decomposition.
Observe that in the case of $\multipushdown{r}{s}$, this is exactly the decomposition into contexts of multi-pushdown systems.

\xsubparagraph{Reductions}. 
Given a computation $w = x_1 \dots x_n$ where each $x_i \in X_\Gamma$, we identify each operation $x_i$ with its position. 
We denote by $w[i]$ the $i$th operation of $w$. 
A \emph{reduction} of $w$ is a finite sequence of applications of the following rewriting rules that transform $w$ into $\varepsilon$. 

\textbf{(R1)} $w'.w[x].w[y].w'' \mapsto_{red}w'w''$, applicable if $w[x]=o^+, w[y]=o^-$ for some $o$.

\textbf{(R2)} $w'.w[x].w[y].w'' \mapsto_{red}w'w''$, applicable if $w[x]=o^-, w[y]=o^+$ for some $oIo$.

\textbf{(R3)} $w'.w[x].w[y].w'' \mapsto_{red}w' w[y]w[x]w''$, applicable if $w[x]Iw[y]$.

Reducing a word $u$ to a word $v$ using these rules is denoted by $u \stackrel{*}{\mapsto_{red}}v$. 
A reduction of $w=a_1 \dots a_n \in X_\Gamma^*$  to $\varepsilon$ is the same as the free reduction of the sequence  $a_1,a_2, \dots, a_n$. 
For any computation $w\in X_\Gamma^*$, we have $w\equiv_\Gamma \varepsilon$ iff $w$ admits a reduction to $\varepsilon$~\cite[Equation (8.2)]{georgthesis}.

Assume that $\pi=w \stackrel{*}{\mapsto_{red}} \varepsilon$ is a reduction that transforms $w$ into $\varepsilon$. 
The relation $R_{\pi}$ relates positions of $w$ which cancel in $\pi$. 
$$w[x] R_{\pi} w[y]~~\text{if}~~w'.w[x].w[y].w'' \mapsto_{red} w'.w''~~\text{or}~~w'.w[y].w[x].w'' \mapsto_{red} w'.w''~~\text{is used in}~~\pi$$

This is inductively lifted to infixes of $w$ as $t_1s_1 R_{\pi} s_2t_2$ if there are contexts $w_i = w_{i1}.t_1.s_1.w_{i2}$ and  
$w_j=w_{j1}.s_2.t_2.w_{j2}$ of $w$ such that $s_1 R_{\pi} s_2$ and $t_1 R_{\pi} t_2$. 

\xsubparagraph{Greedy reductions} A word $w\in X_\Gamma^*$ is called
\emph{irreducible} if neither of the rules \textbf{R1} and \textbf{R2} is
applicable in $w$.  A reduction $\pi\colon w \stackrel{*}{\mapsto_{red}}
\varepsilon$ is called \emph{greedy} if it begins with a sequence of
applications of \textbf{R1} and \textbf{R2} for each context so that the
resulting context is irreducible.  Note that every word $w$ with
$w\equiv_\Gamma\varepsilon$ has a greedy reduction: One can first (greedily)
apply \textbf{R1} and \textbf{R2} until each context is irreducible. Since
the resulting word $w'$ still satisfies $w'\equiv_\Gamma\varepsilon$, there
exists a reduction $w'\stackrel{*}{\mapsto_{red}} \varepsilon$. In total,
this yields a greedy reduction.

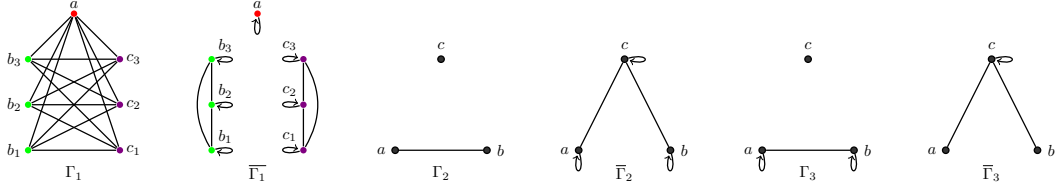
\begin{figure} [t]
  \begin{center}
    \resizebox{\textwidth}{!}{
      \begin{tikzpicture}[-,thick]
        \node[cir] at (-6,-1) (A)
        {} ; 
        
        \node[cir,fill=white] at (-6.3,-1) (AA)
        {$\tiny{b_1}$} ; 
        
        \node[cir1] at (-5,2) (A2)
        {} ;

       \node[cir,fill=white] at (-5,2.2) (AA)
        {$\tiny{a}$} ; 
         
       \node[state, draw=white] at (-5,-1.5) (A0)
        {$\Gamma_1$} ; 

       \node[state, draw=white] at (-1,-1.5) (A0)
        {$\overline{\Gamma_1}$} ; 

        \node[cir2] at (-4,-1) (A1)
        {} ; 
        \node[cir,fill=white] at (-3.7,-1) (AA)
        {$\tiny{c_1}$} ;

        \node[cir] at (-6,0) (B) {};
          \node[cir,fill=white] at (-6.3,0) (AA)
        {$\tiny{b_2}$} ; 
      
        \node[cir2] at (-4,0) (B1) {};
        \node[cir,fill=white] at (-3.7,0) (AA)
        {$\tiny{c_2}$} ;

       \node[cir] at (-6,1) (C) {};
         \node[cir,fill=white] at (-6.3,1) (AA)
        {$\tiny{b_3}$} ; 
      
       \node[cir2] at (-4,1) (C1) {};
       \node[cir,fill=white] at (-3.7,1) (AA)
        {$\tiny{c_3}$} ;

          \path (A) edge node 
        {} (A1);
       \path (A) edge node 
        {} (B1);
         \path (A) edge node 
        {} (C1);
      
      \path (B) edge node 
        {} (A1);
       \path (B) edge node 
        {} (B1);
         \path (B) edge node 
        {} (C1);
      
        \path (C) edge node 
        {} (A1);
       \path (C) edge node 
        {} (B1);
         \path (C) edge node 
        {} (C1);
      
      \path (A2) edge node 
        {} (A);
       \path (A2) edge node 
        {} (B);
         \path (A2) edge node 
        {} (C);
      
      \path (A2) edge node 
        {} (A1);
       \path (A2) edge node 
        {} (B1);
         \path (A2) edge node 
        {} (C1);

          \node[cir] at (-2,-1) (A)
        {} ; 
         \node[cir,fill=white] at (-1.7,-.7) (AA)
        {$\tiny{b_1}$} ; 
       
         \node[cir,fill=white] at (-1.7,0.3) (AA)
        {$\tiny{b_2}$} ; 
       
        \node[cir1] at (-1,2) (A2)
        {} ; 
        \node[cir,fill=white] at (-1,2.2) (AA)
        {$\tiny{a}$} ; 
       
       \node[cir,fill=white] at (-1.7,1.3) (AA)
        {$\tiny{b_3}$} ; 
       
        \node[cir2] at (0,-1) (A1)
        {} ; 
       
       \node[cir,fill=white] at (-0.3,1.3) (AA)
        {$\tiny{c_3}$} ; 
        
        \node[cir] at (-2,0) (B) {};
       \node[cir,fill=white] at (-0.3,0.3) (AA)
        {$\tiny{c_2}$} ; 
       
        \node[cir2] at (0,0) (B1) {};
        \node[cir,fill=white] at (-0.3,-0.7) (AA)
        {$\tiny{c_1}$} ;

       \node[cir] at (-2,1) (C) {};
       \node[cir2] at (0,1) (C1) {};

          \path (A1) edge[loop left] node 
        {} (A1);

          \path (B1) edge[loop left] node 
        {} (B1);
        
          \path (C1) edge[loop left] node 
        {} (C1);
        
          \path (A1) edge node 
        {} (B1);

   \path (A1) edge[bend right] node 
        {} (C1);
          \path (C1) edge node 
        {} (B1);

        \path (A) edge[loop right] node 
        {} (A);

          \path (B) edge[loop right] node 
        {} (B);
        
          \path (C) edge[loop right] node 
        {} (C);
        
          \path (A) edge node 
        {} (B);

   \path (A) edge[bend left] node 
        {} (C);
          \path (C) edge node 
        {} (B);
      
        \path (A2) edge[loop below] node 
        {} (A2);
        
   \node[cirr] at (2,-1) (A)
        {} ; 
   \node[cir,fill=white] at (1.7,-1) (AA)
        {$\tiny{a}$} ; 
            
        \node[state, draw=white] at (3,-1.5) (B0)
        {$\Gamma_2$} ; 
          \node[state, draw=white] at (7,-1.5) (B00)
        {$\overline{\Gamma}_2$} ;

\node[cir,fill=white] at (4.3,-1) (AA)
        {$\tiny{b}$} ;

        \node[cirr] at (3,1) (A2)
        {} ; 
       
       \node[cir,fill=white] at (3,1.3) (AA)
        {$\tiny{c}$} ; 
       
        \node[cirr] at (4,-1) (B1)
        {} ;

   \path (A) edge node 
        {} (B1);

   \node[cirr] at (6,-1) (A)
        {} ; 
           \path (A) edge[loop below] node 
        {} (A);

        \node[cir,fill=white] at (5.7,-1) (AA)
        {$\tiny{a}$} ; 

        \node[cirr] at (7,1) (A2)
        {} ; 
           \path (A2) edge[loop right] node 
        {} (A2);

       \node[cir,fill=white] at (7,1.3) (AA)
        {$\tiny{c}$} ; 
   
        \node[cirr] at (8,-1) (B1)
        {} ; 
             \node[cir,fill=white] at (8.3,-1) (AA)
        {$\tiny{b}$} ; 
              
   \path (A2) edge node 
        {} (B1);
              
   \path (A2) edge node 
        {} (A);
           \path (B1) edge[loop below] node 
        {} (B1);

   \node[cirr] at (10,-1) (A)
        {} ; 
          \path (A) edge[loop below] node 
        {} (A);
  
   \node[cir,fill=white] at (9.7,-1) (AA)
        {$\tiny{a}$} ; 
            
        \node[state, draw=white] at (11,-1.5) (B0)
        {$\Gamma_3$} ; 
          \node[state, draw=white] at (15,-1.5) (B00)
        {$\overline{\Gamma}_3$} ;

\node[cir,fill=white] at (12.3,-1) (AA)
        {$\tiny{b}$} ;

        \node[cirr] at (11,1) (A2)
        {} ;

       \node[cir,fill=white] at (11,1.3) (AA)
        {$\tiny{c}$} ; 
       
        \node[cirr] at (12,-1) (B1)
        {} ; 
              \path (B1) edge[loop below] node 
        {} (B1);

   \path (A) edge node 
        {} (B1);

   \node[cirr] at (14,-1) (A)
        {} ; 
        \node[cir,fill=white] at (13.7,-1) (AA)
        {$\tiny{a}$} ; 

        \node[cirr] at (15,1) (A2)
        {} ; 
       \node[cir,fill=white] at (15,1.3) (AA)
        {$\tiny{c}$} ; 
   
        \node[cirr] at (16,-1) (B1)
        {} ; 
         
             \node[cir,fill=white] at (16.3,-1) (AA)
        {$\tiny{b}$} ; 
              
   \path (A2) edge node 
        {} (B1);
              
   \path (A2) edge node 
        {} (A);

      \path (A2) edge[loop right] node 
        {} (A2);

        \end{tikzpicture}                                                                                                                                                      
      }
      \caption{The storage mechanism of  $\Gamma_1$ is 2 stacks and one partially blind counter. 
      Symbols of the same stack  are weakly dependent. In the storage mechanism of  $\Gamma_2$, $a,b,c$ are weakly dependent. }
      \label{valence}
    \end{center}
  \end{figure}

\xsubparagraph{Weak dependence} In the case of $\Gamma=\multipushdown{r}{s}$, we
know that any two vertices $u,v$ are either dependent (i.e. belong to the same
pushdown) or $\Gamma$ is the direct product of graphs $\Gamma_u$ and $\Gamma_v$
such that $u$ belongs to $\Gamma_u$ and $v$ belongs to $\Gamma_v$. This means,
two operations that are not dependent can, inside every computation, be moved
past each other without changing the effect on the stacks. This is not the case
in general graphs. In $\Gamma_2$ in \cref{valence}, the vertices $a$ and $b$
are not dependent, but in the computation $acb$, they cannot be moved past each
other, because none of them commutes with $c$. We therefore need the additional
notion of weak dependence.
We say that two vertices $u, v \in V$ are \emph{weakly dependent} if there is a
path between them in the complement of the graph.  Here, the \emph{complement}
of a graph $\Gamma=(V, I)$ is obtained by complementing the independence
relation ($v_1 I v_2$ in $\Gamma$ iff we do not have $v_1 I v_2$ in the
complement of $\Gamma$). Equivalently, $u$ and $v$ are not weakly dependent if
$\Gamma$ is the direct product of graphs $\Gamma_u$ and $\Gamma_v$ such that
$u$ belongs to $\Gamma_u$ and $v$ belongs to $\Gamma_v$.  As observed above,
$\Gamma_2$ shows that in general, weakly dependent vertices need not be
dependent.

 It can be seen that weak dependence is an equivalence relation on the set of
 vertices $V$, where the equivalence classes are the connected components in
 the complement of $\Gamma$.  Note that all operations inside a context must
 belong to the same weak dependence class.  We therefore say that two contexts
 $c_1, c_2$ are \emph{weakly dependent} if their operations belong to the same
 weakly dependent equivalence class.  Equivalently, two contexts are weakly
 dependent if all their letters are pairwise weakly dependent.  In particular,
 weak dependency is an equivalence relation on contexts also. Let us denote the
 weak dependence equivalence relation by $\sim_W$ and by $[~]_{\sim_{W}}$ the
 set of all equivalence classes induced by $\sim_W$.

\xsubparagraph{Scope bounded runs}
We now define the notion of bounded scope computations. We first phrase the
classical notion\footnote{The conference version~\cite{concur11} contains a
slightly more restrictive definition. We follow the journal
version~\cite{ic20}.} of scope-boundedness~\cite{ic20} in our framework.  If
$\Gamma=\multipushdown{r}{s}$, then $w\in X_\Gamma^*$ is considered $k$-scope
bounded if there is a reduction $\pi$ for $w$ such that in between any two
symbols $w[i]$ and $w[j]$ related in $R_\pi$, at most $k$ contexts visit the
same anti-clique of $w[i]$ and $w[j]$. Note that in $\multipushdown{r}{s}$,
for every reduction, there is a greedy reduction that induces the same
relation $R_\pi$. Indeed, any applications of \textbf{R1} and \textbf{R2}
that are applicable in a context at the start will eventually be made anyway:
In $\multipushdown{r}{s}$, if a word reduces to $\varepsilon$, then every
position has a uniquely determined ``partner position'' with which is cancels
in every possible reduction.  Therefore, we generalize scope boundedness as
follows. 
\begin{definition}[Scope Bounded Computations] 
	Consider a computation $w \in X_{\Gamma}^+$. We say $w$ is
	\emph{$k$-scoped} if there is a greedy reduction $\pi=w
	\stackrel{*}{\mapsto_{\mathit{red}}} \varepsilon$ such that in between
	any two symbols $w[i]$ and $w[j]$ related by $R_\pi$, at most $k - 1$
	contexts between $w[i]$ and $w[j]$ belong to the same weak dependence
	class as $w[i]$ and $w[j]$. 
 \end{definition}

 By $\mathsf{sc}(w)$, we denote the smallest number $k$ so that $w$ is
 $k$-scoped. Note that there is such a $k$ if and only if $w\equiv_\Gamma
 \varepsilon$. Thus, if $w\not\equiv_\Gamma\varepsilon$, we set
 $\mathsf{sc}(w)=\infty$.
In the example in \cref{valence} (graph $\Gamma_1$) the computation $w=\textcolor{green}{b_1^+}(\textcolor{violet}{c_2^+}\textcolor{red}{a^+} 
\textcolor{violet}{c_1^+}\textcolor{red}{a^+}
\textcolor{violet}{c_1^-}\textcolor{red}{a^-}
\textcolor{violet}{c_2^-}\textcolor{red}{a^-}
)^m \textcolor{green}{b_1^-}$ is 3 scope bounded for all values of $m$, even though the number of context switches grows with $m$.    

\xsubparagraph{Interaction distance}. We make the notion of scope bound more formal using the notion of interaction distance. Given a computation $w \in X_{\Gamma}^+$. Let $c_1c_2 \dots c_n$ be the canonical decomposition of $w$ into contexts. We say that two contexts $c_i, c_j$ with $i<j$ have an \emph{interaction distance} $K$ if there are $K-1$ contexts between $c_i$ and $c_j$ which are weakly dependent with $c_i$. Consider the computation $\textcolor{green}{b_1^+}(\textcolor{red}{a^+}\textcolor{violet}{c_1^+})^{m_1}\textcolor{green}{b_2^+}(\textcolor{red}{a^+}\textcolor{violet}{c_2^+})^{m_2}\textcolor{green}{b_3^+}(\textcolor{red}{a^+}\textcolor{violet}{c_3^+})^{m_3}\textcolor{green}{b_3^-}(\textcolor{violet}{c_3^-}\textcolor{red}{a^-})^{m_3}\textcolor{green}{b_2^-}(\textcolor{violet}{c_2^-}\textcolor{red}{a^-})^{m_2}\textcolor{green}{b_1^-}(\textcolor{violet}{c_1^-}\textcolor{red}{a^-})^{m_1}  $. Each differently colored sequence is a context. 
   The interaction distance 
   between $\textcolor{green}{b_1^+}$ and $\textcolor{green}{b_1^-}$ is 5, since the weakly dependent contexts strictly between them  are  $\textcolor{green}{b_2^+},\textcolor{green}{b_3^+}, \textcolor{green}{b_3^-}, \textcolor{green}{b_2^-}$.

   Thus, $w$ is $k$-scoped if and only if there is a greedy reduction $\pi\colon w\mapsto_{\mathit{red}}\varepsilon$ such that whenever $w[i]R_\pi w[j]$, then the contexts of $w[i]$ and $w[j]$ have interaction distance at most $k$.

The following is the central decision problem studied in this paper.
\begin{center}
\fbox{
\parbox{\textwidth}{
\textbf{The Bounded Scope Reachability Problem}($\BSREACH$)\\
\textbf{Given:} Graph $\Gamma$, scope bound $k$, valence system $\cA$ over $\Gamma$, initial state $q_{\mathit{init}}$, final state $q_{\mathit{fin}}$ \\
\textbf{Decide:} Is there a run from $(q_{\mathit{init}}, \varepsilon)$ to $(q_{\mathit{fin}}, w)$, for some $w\in X_\Gamma^*$ with $\mathsf{sc}(w) \leq k$?
}
}

\end{center}
Thus, in $\BSREACH$,  both
$\Gamma$ and $k$ are part of the input. We also
consider versions where certain parameters are fixed: If
$\Gamma$ is fixed, we denote the problem by
$\BSREACH(\Gamma)$.  If $\Gamma$ is part of the input, but can be
drawn from a class $\cG$ of graphs, we write $\BSREACH(\cG)$.
Finally, if we fix $k$, we use a subscript $k$, resulting
in the problems $\BSREACH_k$, $\BSREACH_k(\Gamma)$, $\BSREACH_k(\cG)$.

Deciding whether there is a run $(q_{\mathit{init}},\varepsilon)$ to
$(q_{\mathit{fin}},w)$ with $w\equiv_\Gamma\varepsilon$ corresponds to  
general configuration reachability~\cite{DBLP:conf/stacs/Zetzsche15}. Hence, we
consider the scope-bounded version of configuration reachability.
	
\xsubparagraph{Strongly Induced Subgraphs}.  When we study decision problems for
valence systems over graph monoids, then typically, if $\Delta$ is an induced
subgraph of $\Gamma$, then a problem instance for $\Delta$ can trivially be
reduced to an instance over $\Gamma$. Here, induced subgraph means that
$\Delta$ can be embedded into $\Gamma$ so that there is an edge in $\Delta$ iff
there is one in $\Gamma$.

This is not necessarily the case for $\BSREACH$: An induced subgraph might
decompose into different weak dependence classes than $\Gamma$. Therefore, we
use a stronger notion of embedding.  We say that $\Gamma'=(V',I')$ is a
\emph{strongly induced subgraph} of $\Gamma=(V,I)$ if there is an injective map
$\iota\colon V\to V'$ such that for any $u,v\in V$, we have (i)~$uIv$ iff
$\iota(u)I'\iota(v)$ and (ii)~$u\sim_W v$ iff $\iota(u) \sim_W \iota(v)$. For
example, the graph $\Gamma$ consisting of two adjacent vertices (without loops)
is an induced subgraph of $\Gamma_2$ in \cref{valence}.  However, $\Gamma$ is
not a strongly induced subgraph of $\Gamma_2$: In $\Gamma_2$, $a$ and $b$ are
weakly dependent, whereas the vertices of $\Gamma$ are not.

\xsubparagraph{Neighbor Antichains}
 Let $\Gamma=(V,I)$ be a graph. In our algorithms, we will need to store information about a dependent set $U\subseteq V$ from which we can conclude whether for another dependent set $U'\subseteq V$, we have $UIU'$; that is, for all $u \in U, u' \in U', u I u'$. 
 To estimate the required information, we use the notion of neighbor antichains.
Let $\Gamma=(V,I)$ be a graph. Given $v \in V$, let $N(v)$ represent the neighbors of $v$, that is $N(v)=\{u \in V \mid uIv\}$.
We define a quasi-ordering on $V$ as follows. For $u,v\in V$, we have $u\le v$ if $N(u)\subseteq N(v)$. 
It is possible that for distinct, $u,v\in V$ we have $u\le v$ and $v\le u$ and thus $\le$ is not necessarily a partial order.
In the following, we will assume that the graphs $\Gamma$ are always equipped with some linear order $\ll$ on $V$. For example, one can just take the order in which the vertices appear in a description of $\Gamma$.
Using $\ll$, we can turn $\le$ into a partial order, which is easier to use algorithmically: We set $u\preceq v$ if and only if $u\le v$ and $u\ll v$. 

Now, given $U \subseteq V$, let $\mini U=\{u \in U \mid \forall v \in U\setminus\{u\}, v \not\preceq u\}$
and $\maxi U=\{u \in U \mid \forall v \in U\setminus\{u\}, u\not\preceq v\}$
 represent respectively, the minimal and  maximal elements from $U$. 

\begin{lemma}
For sets $U,U' \subseteq V$, 	$UIU'$ if and only if $(\mini U)I(\mini U')$.
\label{mini}
\end{lemma}

 Since $\mini U$ and $\mini U'$ are antichains w.r.t.\ $\preceq$, if we bound the size of such antichains in our graph $\Gamma$, we bound the amount of information needed to store to determine whether $UIU'$.
We call a subset $A\subseteq V$ a \emph{neighbor antichain} if (i)~$A$ is dependent (i.e. an anti-clique, no edges between any two vertices of $A$) and (ii)~$A$ is an antichain with respect to $\preceq$.
For the graph  $\Gamma_1$ in Figure \ref{valence}, each vertex is a  neighbor antichain, while for $\overline{\Gamma}_1$, each $\{a,b_i,c_j\}$ is a neighbor antichain for all $i,j$.
By $\tau(\Gamma)$, we denote the maximal size of a neighbor antichain in $\Gamma$.
Thus $\tau(\Gamma_1)=1, \tau(\overline{\Gamma}_1)=3$.
We say that a class $\cG$ of graphs is \emph{neighbor antichain bounded} if there is a number $t$ such that $\tau(\Gamma)\le t$ for every graph $\Gamma$ in $\cG$.

For example,  the class of graphs $\cG$ consisting of bipartite graphs $B_n$ with nodes $\{u_i,v_i\mid i\in\{1,\ldots,n\}\}$, where $\{u_i,v_j\}$ is an edge iff $i\ne j$, is not neighbor antichain bounded. 

\section{Main results}\label{sec:main-results}
In this section, we present the main results of this work.
If both the graph and the scope bound $k$ are part of the input, the bounded
scope reachability problem is $\PSPACE$-complete (as we will show in
\cref{main:k-input}).  Since graph monoids provide a much richer class of
storage mechanisms than multi-pushdowns, this raises the question of
how the complexity is affected if the storage mechanism (i.e. the
graph) is drawn from a subclass of all graphs.  
\begin{theorem}[Scope bound in input]\label{main:k-input}
  Let $\cG$ be a class of graphs.
  Then $\BSREACH(\cG)$ is
  \begin{enumerate}
  \item $\NL$-complete if the graphs in $\cG$ have at most one vertex,
  \item $\P$-complete if every graph in $\cG$ is an anti-clique and $\cG$ contains a graph with $\ge 2$ vertices,
  \item $\PSPACE$-complete otherwise.
  \end{enumerate}
\end{theorem}

\begin{corollary}\label{main:k-input-graph-fixed}
  Let $\Gamma$ be a graph. Then $\BSREACH(\Gamma)$ is
  \begin{enumerate}
  \item $\NL$-complete if $\Gamma$ has at most one vertex,
  \item $\P$-complete if $\Gamma$ is an anti-clique with $\ge 2$ vertices,
  \item $\PSPACE$-complete otherwise.
  \end{enumerate}
\end{corollary}

\subparagraph{Fixed scope bound} We notice that the problem $\BSREACH(\cG)$ is
below $\PSPACE$ only for severely restricted classes $\cG$, where bounded scope
reachability degenerates into ordinary reachability in pushdown automata or
one-counter automata.  Therefore, we also study the setting where the scope
bound $k$ is fixed. However, our result requires two assumptions on the graph
class $\cG$. The first assumption is that $\cG$ be closed under taking strongly
induced subgraphs. This just rules out pathological exceptions: otherwise, it
could be that there are hard instances for $\BSREACH_k$ that only occur
embedded in extremely large graphs in $\cG$, resulting in lower complexity. In
other words, we restrict our attention to the cases where an algorithm for
$\cG$ also has to work for strongly induced subgraphs.  For each individual
graph, this is always the case: if $\Delta$ is a strongly induced subgraph of
$\Gamma$, then $\BSREACH_k(\Delta)$ trivially reduces to $\BSREACH_k(\Gamma)$.

Our second assumption is that $\cG$ be neighbor antichain bounded.  This is a
non-trivial assumption that still covers many interesting types of
infinite-state systems from the literature. For example, every graph mentioned
in \cref{storage-examples} has neighbor antichains of size at most~$1$. In
particular, our result still generalizes the case of multi-pushdown systems.

Moreover, consider the graphs $\stackedcounter{m}$ for $m\in\N$, where
(i)~$\stackedcounter{0}$ is a single unlooped vertex,
(ii)~$\stackedcounter{2m+1}$ is obtained from $\stackedcounter{2m}$ by adding a
new vertex adjacent to all existing vertices, and (iii)~$\stackedcounter{2m+2}$
is obtained from $\stackedcounter{2m+1}$ by adding an isolated unlooped vertex.
Then neighbor antichains in $\stackedcounter{m}$ are of size at most $1$.
Furthermore, using reductions from
\cite[Proposition~3.6]{DBLP:journals/iandc/Zetzsche21}, it follows that
whenever reachability for valence systems over $\Gamma$ is decidable, then this
problem reduces in polynomial time to reachability over some
$\stackedcounter{m}$. Whether reachability is decidable for the graphs
$\stackedcounter{m}$ remains an open
problem~\cite{DBLP:journals/iandc/Zetzsche21}. Thus the graphs
$\stackedcounter{m}$ form an extremely expressive class that is still
neighbor antichain bounded.

\begin{theorem}[Fixed scope bound]\label{main:k-fixed}
  Let $\cG$ be closed under strongly induced
  subgraphs and neighbor antichain bounded.  For every $k\ge 1$, the
  problem $\BSREACH_k(\cG)$ is
  \begin{enumerate}
  \item $\NL$-complete if $\cG$ consists of cliques of bounded size,
  \item $\P$-complete if $\cG$ contains some graph that is not a clique, and the size of cliques in $\cG$ is bounded,
  \item $\PSPACE$-complete otherwise.
  \end{enumerate}
\end{theorem}

In \cref{main:k-fixed}, we do not know if one can lift the restriction of
neighbor antichain boundedness. In \cref{sec:conclusion}, we
describe a class of graphs that is closed under strongly induced
subgraphs, but we do not know the exact complexity of $\BSREACH_k(\cG)$.

\Cref{main:k-fixed} allows us to deduce the complexity of
$\BSREACH_k(\Gamma)$ for every $\Gamma$. 
\begin{corollary}\label{main:k-fixed-graph-fixed}
  Let $\Gamma$ be a graph. Then for every $k\ge 1$, the problem $\BSREACH_k(\Gamma)$ is
  \begin{enumerate}
  \item $\NL$-complete if $\Gamma$ is a clique,
  \item $\P$-complete otherwise.
  \end{enumerate}
\end{corollary}
\begin{proof}
	Apply \cref{main:k-fixed} to the class consisting of $\Gamma$ and its strongly induced subgraphs.
\end{proof}

\subparagraph{Discussion of results} In the case of multi-pushdown systems,
La~Torre, Napoli, and Parlato~\cite{ic20} show that scope-bounded reachability
belongs to $\PSPACE$, and is $\PSPACE$-hard if either the number of stacks or
the scope bound $k$ is part of the input.  Our results complete the picture in
several ways. If $k$ is part of the input, then $\PSPACE$-hardness even holds
if we have two $\Nn$-valued counters instead of stacks (\cref{main:k-input}).
Moreover, hardness also holds when we have two $\Zn$-valued counters (which
often exhibit lower complexities~\cite{DBLP:conf/rp/HaaseH14}).  Moreover, we
determine the complexity the case that both $k$ and the number $s$ of stacks is
fixed.  

Our results can also be interpreted in terms of vector addition systems with
states (VASS).  In the case of VASS (i.e. unlooped cliques), our results imply
that scope-bounded reachability is $\PSPACE$-complete if either (i)~the number
$d$ of counters or (ii)~the scope-bound $k$ are part of the input (and $d\ge
2$).  The same is true if we have integer
VASS~\cite{DBLP:conf/rp/HaaseH14} (looped cliques).  

Thus, for VASS, scope-bounding reduces the complexity of reachability from at
least non-elementary~\cite{DBLP:conf/stoc/CzerwinskiLLLM19} to $\PSPACE$.
Interestingly, for two counters, the complexity goes up from $\NL$ for
general reachability~\cite{DBLP:conf/lics/EnglertLT16} to $\PSPACE$. For
integer VASS, we go up from $\NP$ for general
reachability~\cite{DBLP:conf/rp/HaaseH14} and for a fixed number of
counters even from $\NL$~\cite{GURARI1981220}, to $\PSPACE$.

Note that we obtain a much more complete picture compared to what is known for
bounded context switching~\cite{DBLP:conf/concur/MeyerMZ18}. There, even the
complexity for many individual graphs is not known.  Moreover, the case of
fixed context bounds has not been studied in the case of bounded context switching.

\section{Block decompositions}\label{sec:block-decompositions}
In this section, we lay the foundation for our decision procedure in \cref{sec:decision-procedure}.  We show that in every scope-bounded run $w$, each context can be decomposed into a bounded number of ``blocks'', which will guarantee that $w$ can be reduced to $\varepsilon$ by way of ``block-wise'' reductions. In our algorithms, this will allow us to abstract from each block (which can have unbounded length) by a finite amount of data.
This is similar to the block decomposition in~\cite{DBLP:conf/concur/MeyerMZ18}. 

\begin{figure}
\begin{center}
\includegraphics[scale=.13]{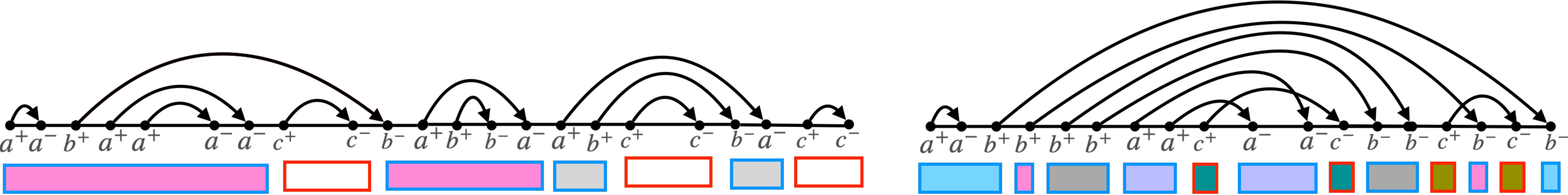}
\end{center}
\caption{Context and block decomposition for a computation over $\overline{\Gamma}_2$ in Figure \ref{valence}. The blue and red lined rectangles are the two contexts. The color filled rectangles represent blocks, with partner  cancelling blocks having the same color filling. }	
\label{decomp}
\end{figure}

Let $w\in X_\Gamma^*$ such that $w\equiv_\Gamma \varepsilon$ with a reduction $\pi$. We call a decomposition $w=w_1\cdots w_m$ a \emph{block decomposition} if it refines the canonical context decomposition\footnote{In other words, each context in $w$ consists of a contiguous subset of the factors $w_1,\ldots,w_m$.} and for each $w_i$, there is a $w_j$ such that
$R_\pi$ relates every position in $w_i$ with a position in $w_i$ itself or in $w_j$.
Here, we do not rule out $i=j$: A block may itself reduce to $\varepsilon$.

\xsubparagraph{Free reductions}
Block decompositions are closely related to free reductions.
Let $w_1, \dots, w_m$ be a sequence of computations in $ X_{\Gamma}^*$. 
A free reduction is a finite sequence of applications of the rewriting rules below to consecutive entries of the sequence so that $w_1, \dots, w_m$ gets transformed into the empty sequence. 

\textbf{(FR1)} $w_i, w_j \mapsto_{\mathit{free}}\varepsilon$ if $w_iw_j\equiv_\Gamma \varepsilon$

\textbf{(FR2)} $w_i, w_j \mapsto_{\mathit{free}} w_j,w_i$ if $w_i I w_j$   

\textbf{(FR3)} $w_i \mapsto_{\mathit{free}} \hat{w}_i$ if $w_i \stackrel{*}{\mapsto_{red}} \hat{w}_i$ using rules \textbf{R1} and \textbf{R2}

\noindent We say that $w_1, \dots, w_m$ is \emph{freely reducible} if it admits a free reduction to the empty sequence. 

As in \cite{DBLP:conf/concur/MeyerMZ18}, we have : 
\begin{restatable}{proposition}{blockDecompositionsFreelyReducible}\label{block-decompositions-freely-reducible}
	If the decomposition $w=w_1\cdots w_m$ refines the context
	decomposition, then it is a block decomposition if and only if the
	sequence $w_1,\ldots,w_m$ is freely reducible.	
\end{restatable}

The main result of this section is that in a scope-bounded run, there exists a
block decomposition with a bounded number of factors in each context.
\begin{restatable}{theorem}{blockDecompositionFewBlocks}\label{block-decomposition-few-blocks}
	Let $w\in X_\Gamma^*$ with $\mathsf{sc}(w)\le k$. Then, there exists a block
	decomposition of $w$ such that each context splits into at most $2k$ blocks.
\end{restatable}
Let us sketch the proof. The \emph{block decomposition} is obtained
by scanning each context $c$ from left to right. As long as there is another
context $c'$ such that all symbols either cancel inside $c$ or with a symbol in
$c'$, we add symbols to the current block.  When we encounter a symbol that
cancels with a position outside of $c$ and $c'$, we start a new block.
This clearly yields a block decomposition (see Figure~\ref{decomp} for an
example) and with arguments similar to \cite{DBLP:conf/concur/MeyerMZ18}, one
can show that it results in at most $2k$ blocks per context. 

\section{Decision procedure}\label{sec:decision-procedure}

In this section, we present the algorithms for the upper bounds of \cref{main:k-input,main:k-fixed}.

%
%
%
%
%
%
%
%
%
%
%
%
%
%
%
%

%
%
%
%
%
%
%
%
%

\xsubparagraph{Block abstractions}\label{babs}
The algorithm for bounded context switching in
\cite{DBLP:conf/concur/MeyerMZ18} abstracts each block by a non-deterministic
automaton. This approach uses polynomial space per block, which would not be
a problem for our $\PSPACE$ algorithm. However, for our $\P$ and $\NL$
algorithms, this would require too much space. Therefore, we begin with a new
notion of ``block abstraction'', which is more space efficient.

Let $w=w_1\cdots w_m$ be a block decomposition for a run of a valence
system $\cA$.  Then it follows from \cref{block-decompositions-freely-reducible} that there are words $\hat{w}_1,\ldots,\hat{w}_m$ such that $w_i\stackrel{*}{\mapsto_{red}} \hat{w}_i$ for each $i$ such that the sequence $\hat{w}_1,\ldots,\hat{w}_m$
can be reduced to the empty sequence using the rewriting rules \textbf{FR1} and \textbf{FR2}.
For each $i$, we store (i)~the states occupied at
the beginning and end of $w_i$, (ii)~its first operation $f\in X_\Gamma$ in $w_i$, (iii)~a non-deterministically chosen operation $o\in X_\Gamma$ occurring in $w_i$, and (iv)~sets $U_i^\mini$ and $U_i^\maxi$,
such that every maximal vertex occurring in $w_i$ is contained in $U_i^\maxi$ and every minimal vertex occurring in $\hat{w}_i$ is contained in $U_i^\mini$.
In this case, we will say
that the block abstraction ``represents'' the word $\hat{w}_i$. Thus, a \emph{block
abstration} is a tuple $N=(q_1,q_2,f,o,U^\mini,U^\maxi)$, where $q_1,q_2$ are
states in $\cA$, $f,o\in X_\Gamma$ are symbols, and $U^\mini,U^\maxi\subseteq V$ are neighbor antichains.
Note that for every set $B\subseteq V$, the sets $\mini B$ and $\maxi B$
are neighbor antichains. 
Formally, we say that $N=(q_1,q_2,f,o,U^\mini,U^\maxi)$ \emph{represents} $\hat{u}\in X_\Gamma^*$
if there is a word $u\in X_\Gamma^*$ such that
(i)~$u\stackrel{*}{\mapsto_{red}} \hat{u}$, (ii)~$u$ is read on some path from $q_1$ to $q_2$, (iii)~$u$ begins with $f$, (iv)~$o$ occurs in $u$, (iv)~the set of vertices 
$B$ occurring in $u$ is a dependent set, (v)~we have $\maxi B\subseteq U^\maxi$,  and
(vi)~$\mini \hat{B}\subseteq U^\mini$, where $\hat{B}$ is the set of vertices occurring in $\hat{u}$. Then, $L(N)$ denotes the set of all words represented by
$N$. We say that two block abstractions $N=(q_1,q_2,f,o,U^\mini,U^\maxi)$ and $N'=(q'_1,q'_2,f',o',U'^\mini,U'^\maxi)$
are \emph{dependent} if $U^\maxi\cup U'^\maxi$ is a dependent set.

\xsubparagraph{Context abstractions} Similarly, we will also
need to abstract contexts.  For this, we need to abstract each of its blocks.
In addition, we need to store the whole context's first symbol ($f$) and some
non-deterministically chosen other symbol ($o$).  These additional symbols will be
used to verify that we correctly guessed the context decomposition of $w$.
Thus, a \emph{context abstraction} is a tuple $\mathcal{C}=(N_1,\ldots,N_{2k},f,o)$, where $N_1,\ldots,N_{2k}$ are pairwise dependent block abstractions, and $f$ and $o$ are symbols.
We say that a context abstraction $\mathcal{C}=(N_1,\ldots,N_{2k},f,o)$ is \emph{independent} with a context abstraction $\mathcal{C'}=(N'_1,\ldots,N'_{2k},f',o')$ if (i)~$f' I o$ and (ii)~for some $i\in\{1,\ldots,2k\}$, the block abstraction $N_{i}=(q^i_1,q^i_2,f_i,o_i,U_i^\mini,U_i^\maxi)$ satisfies $o=o_i$. Intuitively, this means if we have a word represented by $\mathcal{C'}$ and then append a word represented by $\mathcal{C}$, then these words will be the contexts in the canonical context decomposition.

In our algorithms, we will need to check whether a block decomposition admits a
free reduction. This means, we need to check whether the words represented by
block abstractions can cancel (to apply rule FR1) or commute (to apply FR2).
Let us see how to do this. We say that the block abstractions
$N=(q_1,q_2,f,o,U^\mini,U^\maxi)$ and $N'=(q'_1,q'_2,f',o',U'^\mini,U'^\maxi)$
\emph{commute} if $U^\mini I U'^\mini$. Note that if $N$ and $N'$ commute, then
$uIu'$ for every $u\in L(N)$ and $u'\in L(N')$.
We need an analogous condition for cancellation.  We say that $N$ and $N'$
\emph{cancel} if there are words $u\in L(N)$ and $u'\in L(N')$ such that
$uu'\equiv_\Gamma \varepsilon$.
This allows us to define an analogue of free reductions on block abstractions. 
\begin{definition}
	A \emph{free reduction} on a sequence $N_1,\ldots,N_m$ of block abstractions is a sequence of operations

	\noindent	\textbf{(FRA1)} $N_i, N_j \rightarrow_{\mathit{free}} \varepsilon$, if $N_i$ and $N_j$ cancel \\
	\textbf{(FRA2)} $N_i, N_j \rightarrow_{\mathit{free}} N_j, N_i$, if $N_i$ and $N_j$ commute.
	\label{def:fr}
\end{definition}

Together with \cref{mini}, the following \lcnamecref{check-blocks-cancellable} allows us to verify the steps in a free reduction on block abstractions.
\begin{lemma}\label{check-blocks-cancellable}
  Given $\Gamma$, a valence system over $\Gamma$, and
  block abstractions $N_1$ and $N_2$, one can decide in
 $\mathsf{P}$ whether $N_1$ and $N_2$ cancel. If
  $\Gamma$ is a clique, this can be decided in $\NL$.
\end{lemma}
Given block abstractions $N_1$ and $N_2$, (i) first perform saturation \cite{DBLP:conf/concur/MeyerMZ18}, 
obtaining irreducible blocks. Saturation can be implemented using reachability in a one counter automaton, known to be NL-complete \cite{DBLP:conf/time/DemriG07}, (ii) second, construct a pushdown automaton that is non-empty if and only if the saturated $N_1$ and
$N_2$ cancel. Emptiness of pushdown automata is decidable in $\P$. 
If $\Gamma$ is a clique, then the pushdown automaton uses only a
single stack symbol. Hence, we only need a one-counter automaton,
for which emptiness is decidable in $\NL$~\cite{DBLP:conf/time/DemriG07}.
This yields the two upper bounds claimed in \cref{check-blocks-cancellable}.

\xsubparagraph{Reduction to a Reachability Graph} In all our algorithms, we
reduce $\BSREACH$ for  a  valence system $\cA$ over $\Gamma$ to reachability in
a finite graph $\mathcal{R}$. For the $\PSPACE$ algorithm, we argue that each
node of $\mathcal{R}$ requires polynomially many bits and the edge relation can
be decided in $\PSPACE$. For the $\P$ and the $\NL$ algorithm, we argue that
$\mathcal{R}$ is polynomial in size. Moreover, for the $\P$ algorithm, we
compute $\mathcal{R}$ in polynomial time.  For the $\NL$ algorithm, we show
that the edge relation of $\mathcal{R}$ can be decided in $\NL$.

The vertices
 of $\mathcal{R}$ maintain $k$ context abstractions (hence $2k^2$ block abstractions) per weak dependence class. 
 Let $d \leq |V|$ 
be the number of weak dependence classes.   
The idea behind this choice of vertices is to build up a 
 computation $w$ from left to right,  
  by guessing  $k$ context abstractions $\left(\mathcal{C}^{\gamma}_1, \mathcal{C}^{\gamma}_2,\cdots,\mathcal{C}^{\gamma}_k\right)$
per equivalence class $\gamma$ which forms the ``current window'' of $w$. The initial vertex consists of 
$k$ tuples $(\mathcal{E}, \dots, \mathcal{E})$ per weak dependence equivalence class, where 
$\mathcal{E}$ is a placeholder representing a cancelled block or an empty block.  An edge is added  from a vertex $v_1$ 
to a vertex $v_2$ in the graph when the left most context in $v_1$ corresponding to an equivalence class  
$\gamma$ cancels out completely, and $v_2$ is obtained by appending a fresh context abstraction $\mathcal{C}_y^{\gamma}$ to $v_1$. 
Indeed if $w$ is $k$ scope bounded, then the blocks of the first context abstraction
 $\mathcal{C}^\gamma_1$ must cancel out with blocks from the remaining context abstractions $\mathcal{C}^\gamma_2,\cdots,\mathcal{C}^\gamma_k$ using the free reduction 
 rules discussed above. We can guess an equivalence class $\gamma$ whose context abstraction $\mathcal{C}^\gamma_1$ cancels out, and  
 extend $w$ by guessing the next context abstraction $\mathcal{C}_y^\gamma$ in the same equivalence class. An edge between two vertices in $\mathcal{R}$ represents an extension 
 of $w$ where a new context of an appropriate equivalence class is added, after the leftmost  context has cancelled out using 
 some free reduction rules.

 For a weak dependence class, we refer to the tuple of $k$ contexts of interest as a \emph{configuration}. Thus, a vertex in $\mathcal{R}$ consists of $d$ configurations. As discussed earlier in section \ref{babs}, we 
 represent the $2k^2$ blocks in each configuration using block abstractions.

\begin{definition}
\label{configuration-of-classes}
Given a weak dependence class $\gamma \in [~]_{\sim_{W}}$, 	a \emph{configuration of $\gamma$}  is a $k$-tuple of the form
	$s_{\gamma} = (\mathcal{C}^{\gamma}_1, \mathcal{C}^{\gamma}_2,\cdots,\mathcal{C}^{\gamma}_k)$, where
	each $\mathcal{C}^\gamma_c$ for $1\le c\le k$ is a context abstraction.
\end{definition}
As mentioned above, in slight abuse of terminology, in case of cancellation in free reductions, we also allow $\mathcal{E}$ as a placeholder for cancelled contexts.

Intuitively, the configuration tracks the remaining non-cancelled blocks of the last $k$ contexts of this weak dependence class along with their relative positions in their contexts. 

\begin{definition}
	A vertex in the graph $\mathcal{R}$ has the form $(s_{\gamma_1}, \dots, s_{\gamma_d} | i,q)$ where (i)~$\gamma_1, \dots, \gamma_d$ are the distinct equivalence classes in $[~]_{\sim_W}$, (ii)~each $s_{\gamma_j}$ is a configuration, (iii)~$i \in \{1, \cdots, d\}$ is the current weak dependence class we are on, and (iv)~$q$ is the last state occurring in $s_{\gamma_i}$.
\end{definition}
Here, if $s_{\gamma_i}$ consists just of $\mathcal{E}$, then the last condition is satisfied automatically.

\begin{definition}
	\label{one-step-reachable}
	For $\gamma \in [~]_{\sim_W}$, a configuration $s'_{\gamma}$ is \emph{one-step reachable} from a configuration $s_{\gamma}=(\mathcal{C}_1^{\gamma},  \dots, \mathcal{C}_k^{\gamma})$  iff there exists some context abstraction $\mathcal{C}$ and a sequence of free reduction operations on the sequence of block abstractions 
	$$ N^{\gamma 1}_1, \ldots, N^{\gamma 1}_{2k}, N^{\gamma 2}_1, \ldots, N^{\gamma 2}_{2k}, \ldots 
	N^{\gamma k}_1, \ldots, N^{\gamma k}_{2k}, N_1, \ldots, N_{2k} $$
	resulting in the new sequence (placing $\mathcal{E}$ in a position if the block was cancelled due to the free reduction rules;  otherwise we keep the same block abstraction) 
	$$ N'^{\gamma 1}_1, \ldots, N'^{\gamma 1}_{2k}, N'^{\gamma 2}_1, \ldots, N'^{\gamma 2}_{2k}, \ldots N'^{\gamma k}_1, \ldots, N'^{\gamma k}_{2k}, N'^{\gamma (k+1)}_1, \ldots, N'^{\gamma (k+1)}_{2k} $$
	such that $N'^{\gamma 1}_\ell = \mathcal{E}$, for all $\ell \in \{1, \ldots, 2k\}$, and $ \left({\mathcal{C'}^{\gamma}_2}, \mathcal{C'}^{\gamma}_3,\cdots,\mathcal{C'}^{\gamma}_k, \mathcal{C'}^{\gamma}_{k+1}\right) = s'_{\gamma}$,
	where $\mathcal{C'}^{\gamma}_\ell = (N'^{\gamma \ell}_1,  N'^{\gamma \ell}_2, \cdots,N'^{\gamma \ell}_{2k}, f^{\gamma \ell}, o^{\gamma \ell})$ for $\ell\in\{1,\ldots,k+1\}$, and $N_1,\ldots,N_{2k}$ are the block abstractions in $\mathcal{C}$. In this case, we write $s_\gamma\xrightarrow{\mathcal{C}}s'_\gamma$. 
\end{definition}

In short, we can go from $s_\gamma$ to $s'_\gamma$ in one step if we can add some context abstraction to $s_\gamma$ so that using free reduction steps, we can reach $s'_\gamma$ along with clearing the oldest context.

We are now ready to define the edge relation of $\mathcal{R}$. 
\begin{definition}
\label{reachabilit-graph-edge}
There is an edge in $\mathcal{R}$ from a vertex $v=(s_{\gamma_1}, \dots, s_{\gamma_d} | j,q)$ to a vertex  $v'$  
iff there is some $i\in\{1,\ldots,d\}$ and a configuration $s'_{\gamma_i}$ such that (i)~$s_{\gamma_i}\xrightarrow{\mathcal{C'}}s'_{\gamma_i}$ for some context abstraction $\mathcal{C'}$ such that $\mathcal{C'}$ is independent with the last context abstraction $\mathcal{C}$ of $s_{\gamma_j}$ and (ii)~$q$ is the first state in $\mathcal{C'}$ and (iii)~$v'=(s_{\gamma_1}, \dots, s_{\gamma_{i-1}},s'_{\gamma_i},s_{\gamma_{i+1}}, \dots, s_{\gamma_d} | i,q')$, where $q'$ is the last state of $\mathcal{C'}$.
\end{definition}

 Since a context is a maximal dependent sequence, this check suffices to
 guarantee independence of words represented by $\mathcal{C'}$ and $\mathcal{C}$.

As mentioned above, our algorithm reduces scope-bounded reachability to reachability between two nodes in $\mathcal{R}$. Details can be found in the full version.

\xsubparagraph{Complexity}
We turn to the upper bounds in \cref{main:k-input,main:k-fixed}.
Asymptotically, a block abstraction requires $\log|Q| + 2t\cdot \log|V|$ bits,
where $t$ is an upper bound on the size of neighbor antichains.  Per context,
we store $2k$ block abstractions and two symbols. Let $d$ be the number of weak
dependence classes. In a node of $\mathcal{R}$, we store $k$ contexts per weak
dependence class, a number in $\{1,\ldots,d\}$, and a state.    Hence, asymptotically, we
need $M=dk^2(\log |Q| + t\cdot \log|V|)+\log d+\log|Q|$ bits per node of $\mathcal{R}$.
To simulate the free reduction, we only need a constant multiple of this.
We can thus decide reachability in $\mathcal{R}$ in $\PSPACE$.
\begin{proposition}\label{upper-bound-k-input}
	$\BSREACH$ is in $\PSPACE$.
\end{proposition}

We now look at the upper bounds for the first and second cases of Theorem \ref{main:k-fixed}. 

\begin{proposition}\label{upper-bound-k-fixed-cliques-bounded-cliques}
  Let $\cG$ be a class of graphs that is closed under strongly induced
  subgraphs and neighbor antichain bounded.  If $\cG$ consists of cliques of bounded size, then for
  each $k\ge 1$, the problem $\BSREACH_k(\cG)$ belongs to $\NL$.
\end{proposition}
\begin{proof}
	Our assumptions imply that $d\le |V|$, $t$, and $k$ are bounded. Thus
	$M$ is at most logarithmic in the input. Moreover, we can
	simulate free reductions using logarithmic space, because checking
	whether two block abstractions cancel can be done in $\NL$ by
	\cref{check-blocks-cancellable}.
\end{proof}

\begin{proposition}\label{upper-bound-k-fixed-bounded-cliques}
  Let $\cG$ be a class of graphs that is closed under strongly induced
  subgraphs and neighbor antichain bounded.  If the size of cliques in $\cG$ is bounded, then
  for every $k\ge 1$, the problem $\BSREACH_k(\cG)$ belongs to $\P$.
\end{proposition}
\begin{proof}
	First observe that as in \cref{upper-bound-k-fixed-cliques-bounded-cliques},
	the parameters $d$, $t$, and $k$ are bounded. To see this for $d$, let $\ell$
	be an upper bound on the size of cliques in $\cG$. 
	Then, every graph $\Gamma$ in $\cG$ can have at most $\ell$ weak
	dependence classes: Otherwise, $\Gamma$ would have a clique with
	$\ell+1$ nodes as a strongly induced subgraph, and thus $\cG$ would
	contain a clique with $\ell+1$ nodes. Hence, $d$ is bounded and for a
	node in $\mathcal{R}$, we need only logarithmic space.  Moreover, by
	\cref{check-blocks-cancellable}, we can verify a free reduction step in
	$\P$.
\end{proof}

\xsubparagraph{Special Graphs}
We turn to the $\NL$ and $\P$ upper bounds for \cref{main:k-input}. In each
case, all graphs are anti-cliques. Thus, every run has a single context and
$\BSREACH$ reduces to membership for pushdown automata, which is in $\P$. If
there is just one vertex, we can even obtain a one-counter
automaton, for which emptiness is in $\NL$~\cite{DBLP:conf/time/DemriG07}.
\begin{restatable}{proposition}{upperBoundKInputSimple}\label{upper-bound-k-input-simple}
	If $\cG$ is the class of anti-cliques, then $\BSREACH(\cG)$ is in $\P$.
	Moreover, if $\Gamma$ has only one vertex, then $\BSREACH(\Gamma)$ is in $\NL$.
\end{restatable}

\section{Hardness}\label{sec:hardness}
In this section, we show the hardness results of \cref{main:k-input,main:k-fixed}.
\begin{restatable}{proposition}{lowerBoundKFixedTwoNonadjacent}\label{lower-bound-k-fixed-two-nonadjacent}
  If $\Gamma$ is not a clique, then 
  $\BSREACH_k(\Gamma)$ is $\P$-hard for each $k\ge 1$.
\end{restatable}
This uses standard techniques.  If a valence system uses only the two
non-adjacent vertices in $\Gamma$, then scope-bounded reachability is the same
as ordinary reachability.  If both vertices are looped, then this is the
rational subset membership problem for a free group of rank~2, for which
$\P$-hardness was observed in~\cite[Theorem~III.4]{DBLP:conf/lics/HaaseZ19}. If
at least one vertex is unlooped, a standard encoding yields a reduction from
emptiness of pushdown automata.

\xsubparagraph{Two adjacent vertices and bounded queue automata}
Our second hardness proof shows $\PSPACE$-hardness in \cref{main:k-input}.
\begin{proposition}\label{lower-bound-k-input-two-adjacent}
  If $\Gamma$ has two adjacent nodes, then
  $\BSREACH(\Gamma)$ is $\PSPACE$-hard.
\end{proposition}

For the proof of \cref{lower-bound-k-input-two-adjacent}, we employ the model
of bounded queue automata.  A \emph{bounded queue automaton (BQA)} is a tuple
$\cA=(Q,n,T,q_0,q_f)$, where (i)~$Q$ is a finite set of \emph{states},
(ii)~$n\in\N$ is the \emph{queue length}, given in unary, (iii)~$T$ is its set of \emph{transitions},
(iv)~$q_0\in Q$ is its \emph{initial state}, and (iv)~$q_f\subseteq Q$ is its
\emph{final state}. A \emph{configuration} of a BQA is a pair $(q,w)\in
Q\times\{0,1\}^n$. A transition is of the form $(q,x,y,q')$, where $q,q'\in Q$ and $x,y\in\{0,1\}$. We write $(q,w)\autstep(q',w')$ if there is a transition
$(q,x,y,q')$ such that (i)~$w$ has prefix $x$ and (ii)~removing $x$ from the left and
appending $y$ on the right yields $w'$. 
The \emph{reachability problem for BQA} is the following: Given a bounded queue
automaton $(Q,n,T,q_0,q_f)$, is it true that $(q_0,0^n)\autsteps (q_f,0^n)$? It
is straightforward to simulate a linear bounded automaton using a bounded queue
automaton and vice-versa, hence reachability for BQA is $\PSPACE$-complete.

\newcommand{\n}{\|}
\newcommand{\bitzero}{\mathbf{0}}
\newcommand{\bitone}{\mathbf{1}}
\newcommand{\biteps}{\mathbf{e}}

\xsubparagraph{General idea and challenge} Let us first assume that the
nodes $u$ and $v$ in $\Gamma$ are not weakly dependent.  The initial
approach for \cref{lower-bound-k-input-two-adjacent} is to encode the
queue content in the current window of $k=n$ contexts. In each
context, we encode a $0$-bit using an occurrence of the letter $u^+$
that can only be cancelled with a future $u^-$. We call this a
\emph{$0$-context}. Likewise, a $1$-bit is encoded by two occurrences
of $u$, which we call a \emph{$1$-context}. Therefore, we abbreviate
$\bitzero=u^+$ and $\bitone=u^+u^+$.  We also have the right inverses
$\bar{\bitzero}=u^-$ and $\bar{\bitone}=u^-u^-$.  To
start a new context, we multiply $v^+v^-$ and use the
abbreviation $\n=v^+v^-$.
With this encoding, it is easy to check that the oldest context
is a $0$-context: Just multiply $\bar{\bitzero}=u^-$ and then start
a new context using $\n=v^+v^-$. This can only succeed if the
oldest context encodes a $0$: If it had encoded a $1$, there would be
another occurrence of $u^+$ that can never be cancelled.

However, it is not so easy to check that the oldest context is a
$1$-context. One could multiply with
$\bar{\bitone}\n=u^-u^-v^+v^-$, but this can succeed even
if the oldest context is a $0$-context: Indeed, the first occurrence of
$u^-$ can cancel with the $u^+$ in the oldest context $c$, but
the second $u^-$ could cancel with $u^+$ in a context to the right of $c$.

\xsubparagraph{Solution}
We overcome this as follows. Instead of one context per bit, we
use three contexts. To encode a $0$ in the queue, we
use a $0$-context, a $1$-context, and another $1$-context, resulting
in the string $011$. To encode a $1$, we do the same with the bit
string $100$.  Then, we use the above approach to check
for $011$ or $100$: Since a successful check for a $0$-context
guarantees that there was a $0$-context, checking for $0,1,1$
guarantees that the oldest context is a $0$-context, thus the three
oldest contexts must carry $011$. When we check for $1,0,0$, then 
among the three oldest contexts, at least two are
$0$-contexts, hence the three oldest contexts carry $100$.

Let us calculate the required scope bound to implement this idea. Since we
only use the operations $u^+,u^-,v^+,v^-$, we only have $u$-contexts
(consisting of $u^+,u^-$) and $v$-contexts (consisting of $v^+,v^-$).  When
we read the oldest bit in the queue, we produce three new $u$-contexts
(separated by $v$-contexts).  Then, we need to write a new bit in the queue,
which requires another three $u$-contexts (separated by $v$-contexts).
The interaction distance to the oldest $u$-context that is part of
the leftmost queue bit is always $k=3(2n-1)$: The oldest bit is encoded using
three $u$-contexts.  For each further queue entry $i\in\{2,\ldots,n\}$, we
have three $u$-contexts that were used to read an even older bit, and then
three $u$-contexts that encode the $i$-th bit in the queue. In total, this yields $3(1+2(n-1))=3(2n-1)$ many $u$-contexts.

To initialize the queue, we use 
$t_0=(\bitzero\n\bitone\n\bitone\n)(\biteps\n\biteps\n\biteps\n\bitzero\n\bitone\n\bitone\n)^{n-1}$. Here, $\biteps=u^+u^-$ is a ``gap context'' that ensures that the leftmost bit has interaction distance exactly $k=3(2n-1)$ from the right end of $t_0$. 
Thus, $t_0$ puts $n$ copies of the
bit string $011$, plus $n-1$ gap contexts into our window of $k=3(2n-1)$ contexts. To simulate a
transition $(p,x,y,p')$, we check that $x$ is the bit encoded
by the three oldest contexts.
Afterwards, we put the new bit $y$ into the queue.
Therefore,  if $x=0$,
define the triple
$(x_1,x_2,x_3)=(\bar{\bitzero},\bar{\bitone},\bar{\bitone})$; if
$x=1$, let
$(x_1,x_2,x_3)=(\bar{\bitone},
\bar{\bitzero},\bar{\bitzero})$. Moreover, if $y=0$, then let
$(y_1,y_2,y_3)=(\bitzero, \bitone,\bitone)$; if $y=1$, then let
$(y_1,y_2,y_3)=(\bitone,\bitzero,\bitzero)$. Then we use the string
$t_{x,y}=x_1\n x_2 \n x_3 \n y_1\n y_2\n y_3\n$. Finally, to check that the
encoded queue content consists entirely of $0$'s, we use 
$t_f=(\bar{\bitzero}\n\bar{\bitone}\n\bar{\bitone}\n)(\biteps\n\biteps\n\biteps\n\bar{\bitzero}\n\bar{\bitone}\n\bar{\bitone}\n)^{n-1}$.
With this encoding, it is straightforward to translate a BQA into a valence
system over $\Gamma$. 

Note that if $u$ and $v$ are weakly dependent, then the same construction
works, except that we have to set $k=6(2n-1)$, because now the $v$-contexts
$\n=v^+v^-$ between two $u$-contexts count towards the interaction distance.  

\xsubparagraph{Unbounded cliques and bit vector automata}
We turn to $\PSPACE$-hardness in \cref{main:k-fixed}.
\vspace{-0.5cm}
\begin{proposition}\label{lower-bound-k-fixed-unbounded-cliques}
  Suppose $\cG$ be either the class of unlooped cliques or the class of
  looped cliques. Then for every $k\ge 1$, the problem
  $\BSREACH_k(\cG)$ is $\PSPACE$-hard.
\end{proposition}

Here it is convenient to reduce from bit vector automata, whose configuration
consists of a state and a bit vector. In each step, they can read and modify
one of the bits.  A \emph{bit vector automaton (BVA)} is a tuple
$(Q,n,T,q_0,q_f)$, where (i)~$Q$ is a finite set of states, (ii)~$n$ is the
\emph{vector length}, given in unary, (iii)~a set $T$ of \emph{transitions},
(iv)~$q_0\in Q$ is its \emph{initial state}, and (v)~$q_f\in Q$ is its
\emph{final state}.  A transition is of the form $(p,i,x,y,q)$ with $p,q\in Q$,
$i\in\{0,\ldots,n\}$, and $x,y\in\{0,1\}$.  It checks that $i$-th bit is
currently $x$, and sets the $i$-th bit to $y$.  Thus, a \emph{configuration} of
a bit vector automaton is a pair $(q,w)\in Q\times\{0,1\}^n$. By $\autsteps$,
we denote the reachability relation. The \emph{reachability problem for BVA}
asks, given a BVA $(Q,n,T,q_0,q_f)$, is it true that $(q_0,0^n)\autsteps
(q_f,0^n)$?  Again, a simulation of linear bounded automata is straightforward
and this problem is $\PSPACE$-complete.

\begin{proof}[Proof of \cref{lower-bound-k-fixed-unbounded-cliques}]
  Let $\cA=(Q,n,T,q_0,q_f)$ be a BVA. Moreover, depending on whether $\cG$ is
  the class of looped or unlooped cliques, let $\Gamma=(V,I)$ be either a
  looped or an unlooped clique with $2n$ vertices, so
  let $V=\{a_i,b_i\mid i\in\{1,\ldots,n\}\}$. Our
  construction does not depend on whether $\Gamma$ is looped or unlooped and we
  will show that it is correct in either case.
  We first illustrate the idea for maintaining a single bit using the vertices $a_i,b_i$. To ease notation, we now write $v$ for $v^+$ and $\bar{v}$ for $v^-$ when $v\in V$.  Consider the string
  \[ w=(a_i^{r_1}b_i\bar{b}_i\bar{a}_i^{s_1}b_i\bar{b}_i)(a_i\bar{a}_ib_i\bar{b}_i)^k (a_i^{r_2}b_i\bar{b}_i\bar{a}_i^{s_2}b_i\bar{b}_i)\cdots (a_i\bar{a}_ib_i\bar{b}_i)^k (a_i^{r_m}b_i\bar{b}_i\bar{a}_i^{s_m}b_i\bar{b}_i). \]
  Moreover, assume that for each $j=1,\ldots,m$, we have $r_j,s_j\in\{k,3k\}$.
  We think of $a_i^{r_j}b_i\bar{b}_i$ as an operation that stores $0$ if $r_j=k$ and stores $1$ if $r_j=3k$.
  We think of $\bar{a}_i^{s_j}b_i\bar{b}_i$ as a read operation, where again $s_j=k$ stands for 0 and $s_j=3k$ stands for $1$. Here, the purpose of $b_i\bar{b}_i$ is to start a new context (since $\Gamma$ is a clique).
Moreover, each factor $(a_i\bar{a}_ib_i\bar{b}_i)^k$ produces $k$ contexts in the weak dependence class of $a_i$, where each context contains $a_i\bar{a}_i$.
  This means, each factor $(a_i\bar{a}_ib_i\bar{b}_i)^k$ enforces an interaction distance of $k+1$ between $\bar{a}_i^{s_j}$ and $a_i^{r_{j+1}}$, and thus prevents them from canceling with each other.
  
  We claim that $\mathsf{sc}(w)\le k$ if and only if each read operation reads the bit that was stored before.
  In other words, we have $\mathsf{sc}(w)\le k$ if and only if $r_j=s_j$ for each $j\in\{1,\ldots,m\}$. Moreover, this is true regardless of whether $\Gamma$ is looped or unlooped.
  For the ``if'', note that each $a_i^{r_j}$ can cancel with $\bar{a}_i^{s_j}$ and in every other context, every letter ($a_i$, $\bar{a}_i$, $b_i$, $\bar{b}_i$) can cancel with its direct neighbor.
  Conversely, suppose $r_j\ne s_j$ for some $j$.
  If $r_j=3k$ and $s_j=k$, then the context $a_i^{r_j}=a_i^{3k}$ sees only $3k-1$ occurrences of $\bar{a}_i$ in contexts at interaction distance $\le k$:
  First, the context $\bar{a}_i^{s_i}=\bar{a}_i^k$ yields $k$ occurrences.
  The other $2k-1$ contexts are of the form $a_i\bar{a}_i$ and each provides one occurrence of $\bar{a}_i$.
  In total, we have $k+2k-1=3k-1$.
  It is thus impossible to cancel every letter in $a_i^{r_j}$. The case $r_j=k$ and $s_j=3k$ is symmetric.
  This proves our claim.
  Using this encoding, it is now straightforward to simulate $n$ bits.
\end{proof}

\Cref{lower-bound-k-fixed-two-nonadjacent,lower-bound-k-input-two-adjacent,lower-bound-k-fixed-unbounded-cliques}
are the last remaining ingredients for \cref{main:k-input,main:k-fixed}:
\Cref{main:k-input} follows from
\cref{upper-bound-k-input,upper-bound-k-input-simple,lower-bound-k-fixed-two-nonadjacent,lower-bound-k-input-two-adjacent}.
Moreover, \cref{main:k-fixed} follows from
\cref{upper-bound-k-input,upper-bound-k-fixed-cliques-bounded-cliques,upper-bound-k-fixed-bounded-cliques,lower-bound-k-fixed-two-nonadjacent,lower-bound-k-fixed-unbounded-cliques}. 

\section{Conclusion}\label{sec:conclusion}
We have introduced a notion of scope-bounded reachability for valence systems
over graph monoids. In the special case of graphs that correspond to
multi-pushdowns, this notion coincides with the original notion of
scope-boundedness introduced by La Torre, Napoli, and Parlato~\cite{ic20}.  We
have shown that with this notion, scope-bounded reachability is decidable in
$\PSPACE$, even if the graph and the scope bound $k$ are part of the input.

In addition, we have studied the complexity of the problem under four types of
restrictions: (i)~$k$ and the graph are part of the input, and the graph is
drawn from some class $\cG$ of graphs, (\cref{main:k-input}), (ii)~$k$ is part
of the input and the graph is fixed (\cref{main:k-input-graph-fixed}),
(iii)~$k$ is fixed and the graph is drawn from some class $\cG$ of graphs that
is neighbor antichain bounded and closed under strongly induced subgraphs
(\cref{main:k-fixed}) and (iv)~$k$ is fixed and the graph is fixed
(\cref{main:k-fixed-graph-fixed}).  We have completely determined the
complexity landscape in the situations (i)--(iv): In every case, we obtain
$\NL$-, $\P$-, or $\PSPACE$-completeness.

These results settle the complexity of scope-bounded reachability for most
types of infinite-state systems that fit in the framework of valence systems
and have been considered in the literature.

\subparagraph{Open Problem: Dropping neighbor antichain boundedness} However,
we leave open what complexities can arise if in case (iii) above, we drop the
assumption of neighbor antichain boundedness. In other words: $k$ is fixed and
the graph comes from a class $\cG$ that is closed under strongly induced
subgraphs.  For all we know, it is possible that there are classes $\cG$ for
which the problem is neither $\NL$-, nor $\P$-, nor $\PSPACE$-complete.

For example, for each $n\ge 0$, consider  the bipartite graph $B_n$ with nodes
$\{u_i,v_i\mid i\in\{1,\ldots,n\}\}$, where $\{u_i,v_j\}$ is an edge if and
only if $i\ne j$.  Moreover, let $\cG$ be the class of graphs containing $B_n$
for every $n\in\N$ and all strongly induced subgraphs.  Observe that the
cliques in $\cG$ have size at most 2: $B_n$ is bipartite and thus every clique
in $B_n$ has size at most $2$. In particular, if a clique is a strongly induced
subgraph of $B_n$, then it can be of size at most $2$.  Moreover, the graphs
$B_n$ have neighbor antichains of unbounded size: The set $\{u_1,\ldots,u_n\}$
is a neighbor antichain in $B_n$. 

We currently do not know the exact complexity of $\BSREACH_k(\cG)$. By
\cref{main:k-fixed}, the problem is $\P$-hard and in $\PSPACE$.  Intuitively,
our $\P$ upper bound does not apply because in each node of $\mathcal{R}$, one
would have to remember $n$ bits in order to keep enough information about
commutation of blocks: For a subset $S\subseteq\{1,\ldots,n\}$, let $u_S$ be
the product of all $u_1^+,\ldots,u_n^+$, where we only include $u_i^+$ if $i\in
S$.  Then $u_S v_j^+\equiv v_j^+ u_S$ if and only if $j\notin S$.

\bibliography{ref.bib}

\begin{thebibliography}{10}

\bibitem{DBLP:conf/atva/AiswaryaGK14}
C.~Aiswarya, Paul Gastin, and K.~Narayan Kumar.
\newblock Verifying communicating multi-pushdown systems via split-width.
\newblock In {\em Automated Technology for Verification and Analysis - 12th
  International Symposium, {ATVA} 2014, Sydney, NSW, Australia, November 3-7,
  2014, Proceedings}, volume 8837 of {\em Lecture Notes in Computer Science},
  pages 1--17. Springer, 2014.
\newblock \href {http://dx.doi.org/10.1007/978-3-319-11936-6\_1}
  {\path{doi:10.1007/978-3-319-11936-6\_1}}.

\bibitem{DBLP:conf/tacas/AkshayGKR20}
S.~Akshay, Paul Gastin, Shankara~Narayanan Krishna, and Sparsa Roychowdhury.
\newblock Revisiting underapproximate reachability for multipushdown systems.
\newblock In {\em Tools and Algorithms for the Construction and Analysis of
  Systems - 26th International Conference, {TACAS} 2020, Held as Part of the
  European Joint Conferences on Theory and Practice of Software, {ETAPS} 2020,
  Dublin, Ireland, April 25-30, 2020, Proceedings, Part {I}}, volume 12078 of
  {\em Lecture Notes in Computer Science}, pages 387--404. Springer, 2020.
\newblock \href {http://dx.doi.org/10.1007/978-3-030-45190-5\_21}
  {\path{doi:10.1007/978-3-030-45190-5\_21}}.

\bibitem{DBLP:conf/atva/AtigBKS12}
Mohamed~Faouzi Atig, Ahmed Bouajjani, K.~Narayan Kumar, and Prakash Saivasan.
\newblock Linear-time model-checking for multithreaded programs under
  scope-bounding.
\newblock In {\em Automated Technology for Verification and Analysis - 10th
  International Symposium, {ATVA} 2012, Thiruvananthapuram, India, October 3-6,
  2012. Proceedings}, volume 7561 of {\em Lecture Notes in Computer Science},
  pages 152--166. Springer, 2012.
\newblock \href {http://dx.doi.org/10.1007/978-3-642-33386-6\_13}
  {\path{doi:10.1007/978-3-642-33386-6\_13}}.

\bibitem{DBLP:conf/dlt/BhaveKPT19}
Devendra Bhave, Shankara~Narayanan Krishna, Ramchandra Phawade, and Ashutosh
  Trivedi.
\newblock On timed scope-bounded context-sensitive languages.
\newblock In {\em Developments in Language Theory - 23rd International
  Conference, {DLT} 2019, Warsaw, Poland, August 5-9, 2019, Proceedings},
  volume 11647 of {\em Lecture Notes in Computer Science}, pages 168--181.
  Springer, 2019.
\newblock \href {http://dx.doi.org/10.1007/978-3-030-24886-4\_12}
  {\path{doi:10.1007/978-3-030-24886-4\_12}}.

\bibitem{DBLP:conf/mfcs/BuckheisterZ13}
P.~Buckheister and Georg Zetzsche.
\newblock Semilinearity and context-freeness of languages accepted by valence
  automata.
\newblock In {\em Mathematical Foundations of Computer Science 2013 - 38th
  International Symposium, {MFCS} 2013, Klosterneuburg, Austria, August 26-30,
  2013. Proceedings}, volume 8087 of {\em Lecture Notes in Computer Science},
  pages 231--242. Springer, 2013.
\newblock \href {http://dx.doi.org/10.1007/978-3-642-40313-2\_22}
  {\path{doi:10.1007/978-3-642-40313-2\_22}}.

\bibitem{DBLP:phd/hal/Cyriac14}
Aiswarya Cyriac.
\newblock {\em Verification of communicating recursive programs via
  split-width. (V{\'{e}}rification de programmes r{\'{e}}cursifs et
  communicants via split-width)}.
\newblock PhD thesis, {\'{E}}cole normale sup{\'{e}}rieure de Cachan, France,
  2014.
\newblock URL: \url{https://tel.archives-ouvertes.fr/tel-01015561}.

\bibitem{DBLP:conf/concur/CyriacGK12}
Aiswarya Cyriac, Paul Gastin, and K.~Narayan Kumar.
\newblock {MSO} decidability of multi-pushdown systems via split-width.
\newblock In {\em {CONCUR} 2012 - Concurrency Theory - 23rd International
  Conference, {CONCUR} 2012, Newcastle upon Tyne, UK, September 4-7, 2012.
  Proceedings}, volume 7454 of {\em Lecture Notes in Computer Science}, pages
  547--561. Springer, 2012.
\newblock \href {http://dx.doi.org/10.1007/978-3-642-32940-1\_38}
  {\path{doi:10.1007/978-3-642-32940-1\_38}}.

\bibitem{DBLP:conf/stoc/CzerwinskiLLLM19}
Wojciech Czerwinski, Slawomir Lasota, Ranko Lazic, J{\'{e}}r{\^{o}}me Leroux,
  and Filip Mazowiecki.
\newblock The reachability problem for petri nets is not elementary.
\newblock In {\em Proceedings of STOC 2019}, pages 24--33. {ACM}, 2019.
\newblock \href {http://dx.doi.org/10.1145/3313276.3316369}
  {\path{doi:10.1145/3313276.3316369}}.

\bibitem{DBLP:conf/time/DemriG07}
St{\'{e}}phane Demri and R{\'{e}}gis Gascon.
\newblock The effects of bounding syntactic resources on presburger {LTL}.
\newblock In {\em Proceedings of TIME 2007}, pages 94--104. {IEEE} Computer
  Society, 2007.
\newblock \href {http://dx.doi.org/10.1109/TIME.2007.63}
  {\path{doi:10.1109/TIME.2007.63}}.

\bibitem{DBLP:conf/lics/DOsualdoMZ16}
Emanuele D'Osualdo, Roland Meyer, and Georg Zetzsche.
\newblock First-order logic with reachability for infinite-state systems.
\newblock In {\em Proceedings of LICS 2016}, pages 457--466. {ACM}, 2016.
\newblock \href {http://dx.doi.org/10.1145/2933575.2934552}
  {\path{doi:10.1145/2933575.2934552}}.

\bibitem{DBLP:conf/lics/EnglertLT16}
Matthias Englert, Ranko Lazic, and Patrick Totzke.
\newblock Reachability in two-dimensional unary vector addition systems with
  states is nl-complete.
\newblock In {\em Proceedings of LICS~2016}, pages 477--484. {ACM}, 2016.
\newblock \href {http://dx.doi.org/10.1145/2933575.2933577}
  {\path{doi:10.1145/2933575.2933577}}.

\bibitem{DBLP:journals/toplas/EsparzaGP14}
Javier Esparza, Pierre Ganty, and Tom{\'{a}}s Poch.
\newblock Pattern-based verification for multithreaded programs.
\newblock {\em {ACM} Trans. Program. Lang. Syst.}, 36(3):9:1--9:29, 2014.
\newblock \href {http://dx.doi.org/10.1145/2629644}
  {\path{doi:10.1145/2629644}}.

\bibitem{GURARI1981220}
Eitan~M. Gurari and Oscar~H. Ibarra.
\newblock The complexity of decision problems for finite-turn multicounter
  machines.
\newblock {\em Journal of Computer and System Sciences}, 22(2):220--229, 1981.
\newblock \href
  {http://dx.doi.org/https://doi.org/10.1016/0022-0000(81)90028-3}
  {\path{doi:https://doi.org/10.1016/0022-0000(81)90028-3}}.

\bibitem{DBLP:conf/rp/HaaseH14}
Christoph Haase and Simon Halfon.
\newblock Integer vector addition systems with states.
\newblock In {\em Proceedings of RP~2014}, volume 8762 of {\em Lecture Notes in
  Computer Science}, pages 112--124. Springer, 2014.
\newblock \href {http://dx.doi.org/10.1007/978-3-319-11439-2\_9}
  {\path{doi:10.1007/978-3-319-11439-2\_9}}.

\bibitem{DBLP:conf/lics/HaaseZ19}
Christoph Haase and Georg Zetzsche.
\newblock Presburger arithmetic with stars, rational subsets of graph groups,
  and nested zero tests.
\newblock In {\em Proceedings of LICS 2019}, pages 1--14. {IEEE}, 2019.
\newblock \href {http://dx.doi.org/10.1109/LICS.2019.8785850}
  {\path{doi:10.1109/LICS.2019.8785850}}.

\bibitem{haddad2007recursive}
Serge Haddad and Denis Poitrenaud.
\newblock Recursive {Petri} nets.
\newblock {\em Acta Informatica}, 44(7):463--508, 2007.

\bibitem{DBLP:conf/icalp/LerouxST15}
J{\'{e}}r{\^{o}}me Leroux, Gr{\'{e}}goire Sutre, and Patrick Totzke.
\newblock On the coverability problem for pushdown vector addition systems in
  one dimension.
\newblock In {\em Proceedings of ICALP 2015}, volume 9135 of {\em Lecture Notes
  in Computer Science}, pages 324--336. Springer, 2015.
\newblock \href {http://dx.doi.org/10.1007/978-3-662-47666-6\_26}
  {\path{doi:10.1007/978-3-662-47666-6\_26}}.

\bibitem{Lohrey2013a}
Markus Lohrey.
\newblock The rational subset membership problem for groups: a survey.
\newblock In {\em Groups St Andrews 2013}, London Mathematical Society Lecture
  Notes Series, pages 368--379. Cambridge University Press, 2016.

\bibitem{DBLP:conf/popl/MadhusudanP11}
P.~Madhusudan and Gennaro Parlato.
\newblock The tree width of auxiliary storage.
\newblock In {\em Proceedings of POPL 2011}, pages 283--294. {ACM}, 2011.
\newblock \href {http://dx.doi.org/10.1145/1926385.1926419}
  {\path{doi:10.1145/1926385.1926419}}.

\bibitem{DBLP:conf/concur/MeyerMZ18}
Roland Meyer, Sebastian Muskalla, and Georg Zetzsche.
\newblock Bounded context switching for valence systems.
\newblock In {\em 29th International Conference on Concurrency Theory, {CONCUR}
  2018, September 4-7, 2018, Beijing, China}, volume 118 of {\em LIPIcs}, pages
  12:1--12:18. Schloss Dagstuhl - Leibniz-Zentrum f{\"{u}}r Informatik, 2018.
\newblock \href {http://dx.doi.org/10.4230/LIPIcs.CONCUR.2018.12}
  {\path{doi:10.4230/LIPIcs.CONCUR.2018.12}}.

\bibitem{qadeer2005context}
Shaz Qadeer and Jakob Rehof.
\newblock Context-bounded model checking of concurrent software.
\newblock In {\em Proceedings of TACAS 2005}, pages 93--107. Springer, 2005.

\bibitem{ramalingam2000context}
Ganesan Ramalingam.
\newblock Context-sensitive synchronization-sensitive analysis is undecidable.
\newblock {\em ACM Transactions on Programming languages and Systems (TOPLAS)},
  22(2):416--430, 2000.

\bibitem{Sippu1988}
Seppo Sippu and Eljas Soisalon-Soininen.
\newblock Springer Berlin Heidelberg, Berlin, Heidelberg, 1988.
\newblock \href {http://dx.doi.org/10.1007/978-3-642-61345-6_5}
  {\path{doi:10.1007/978-3-642-61345-6_5}}.

\bibitem{concur11}
Salvatore~La Torre and Margherita Napoli.
\newblock Reachability of multistack pushdown systems with scope-bounded
  matching relations.
\newblock In {\em {CONCUR} 2011 - Concurrency Theory - 22nd International
  Conference, {CONCUR} 2011, Aachen, Germany, September 6-9, 2011.
  Proceedings}, volume 6901 of {\em Lecture Notes in Computer Science}, pages
  203--218. Springer, 2011.
\newblock \href {http://dx.doi.org/10.1007/978-3-642-23217-6\_14}
  {\path{doi:10.1007/978-3-642-23217-6\_14}}.

\bibitem{tcs12}
Salvatore~La Torre and Margherita Napoli.
\newblock A temporal logic for multi-threaded programs.
\newblock In {\em Theoretical Computer Science - 7th {IFIP} {TC} 1/WG 2.2
  International Conference, {TCS} 2012, Amsterdam, The Netherlands, September
  26-28, 2012. Proceedings}, volume 7604 of {\em Lecture Notes in Computer
  Science}, pages 225--239. Springer, 2012.
\newblock \href {http://dx.doi.org/10.1007/978-3-642-33475-7\_16}
  {\path{doi:10.1007/978-3-642-33475-7\_16}}.

\bibitem{DBLP:journals/ijfcs/TorreNP16}
Salvatore~La Torre, Margherita Napoli, and Gennaro Parlato.
\newblock Scope-bounded pushdown languages.
\newblock {\em Int. J. Found. Comput. Sci.}, 27(2):215--234, 2016.
\newblock \href {http://dx.doi.org/10.1142/S0129054116400074}
  {\path{doi:10.1142/S0129054116400074}}.

\bibitem{ic20}
Salvatore~La Torre, Margherita Napoli, and Gennaro Parlato.
\newblock Reachability of scope-bounded multistack pushdown systems.
\newblock {\em Inf. Comput.}, 275:104588, 2020.
\newblock \href {http://dx.doi.org/10.1016/j.ic.2020.104588}
  {\path{doi:10.1016/j.ic.2020.104588}}.

\bibitem{DBLP:conf/fsttcs/TorreP12}
Salvatore~La Torre and Gennaro Parlato.
\newblock Scope-bounded multistack pushdown systems: Fixed-point,
  sequentialization, and tree-width.
\newblock In {\em {IARCS} Annual Conference on Foundations of Software
  Technology and Theoretical Computer Science, {FSTTCS} 2012, December 15-17,
  2012, Hyderabad, India}, volume~18 of {\em LIPIcs}, pages 173--184. Schloss
  Dagstuhl - Leibniz-Zentrum f{\"{u}}r Informatik, 2012.
\newblock \href {http://dx.doi.org/10.4230/LIPIcs.FSTTCS.2012.173}
  {\path{doi:10.4230/LIPIcs.FSTTCS.2012.173}}.

\bibitem{DBLP:conf/icalp/Zetzsche13}
Georg Zetzsche.
\newblock Silent transitions in automata with storage.
\newblock In {\em Proceedings of ICALP 2013}, volume 7966 of {\em Lecture Notes
  in Computer Science}, pages 434--445. Springer, 2013.
\newblock \href {http://dx.doi.org/10.1007/978-3-642-39212-2\_39}
  {\path{doi:10.1007/978-3-642-39212-2\_39}}.

\bibitem{DBLP:conf/stacs/Zetzsche15}
Georg Zetzsche.
\newblock Computing downward closures for stacked counter automata.
\newblock In {\em Proceedings of STACS 2015}, volume~30 of {\em LIPIcs}, pages
  743--756. Schloss Dagstuhl - Leibniz-Zentrum f{\"{u}}r Informatik, 2015.
\newblock \href {http://dx.doi.org/10.4230/LIPIcs.STACS.2015.743}
  {\path{doi:10.4230/LIPIcs.STACS.2015.743}}.

\bibitem{georgthesis}
Georg Zetzsche.
\newblock {\em Monoids as Storage Mechanisms}.
\newblock PhD thesis, Kaiserslautern University of Technology, Germany, 2016.
\newblock URL:
  \url{https://kluedo.ub.uni-kl.de/frontdoor/index/index/docId/4400}.

\bibitem{DBLP:journals/eatcs/Zetzsche16}
Georg Zetzsche.
\newblock Monoids as storage mechanisms.
\newblock {\em Bull. {EATCS}}, 120, 2016.
\newblock URL: \url{http://eatcs.org/beatcs/index.php/beatcs/article/view/459}.

\bibitem{DBLP:journals/iandc/Zetzsche21}
Georg Zetzsche.
\newblock The emptiness problem for valence automata over graph monoids.
\newblock {\em Inf. Comput.}, 277:104583, 2021.
\newblock \href {http://dx.doi.org/10.1016/j.ic.2020.104583}
  {\path{doi:10.1016/j.ic.2020.104583}}.

\end{thebibliography}

\appendix

\section{Additional material for Section~\ref{sec:scope-bounded}}\label{sec:appendix-scope-bounded}

\subparagraph{Proof of Lemma \ref{mini}}
\begin{proof}
Assume $UIU'$. That is, for all $u \in U$, $u' \in U'$, $uIu'$. Assume that $\neg[(\mini U)I(\mini U')]$. 
Then there is some $v \in \mini(U) \subseteq U, v' \in \min(U')\subseteq U'$ such that $v, v'$ are dependent. This contradicts 
 $UIU'$. Likewise, if $(\mini U)I(\mini U')$ then $u I u'$ for all $u \in \mini(U),u' \in \mini(U')$. Assume $\neg(U I U')$. 
  Consider $u \prec v, v \in U$ and $u' \prec v', v' \in U'$ such that $\neg(v I v')$. The last condition is possible by assumption that  $\neg(U I U')$. 
     Then we have $N(u) \subset N(v), N(u') \subset N(v')$ but $\neg(v I v')$.
 Since $u I u'$ and $u \prec v, u' \prec v'$,  we have $v' \in N(u)$ and $v \in N(u')$. Hence, $v' \in N(v), v\in N(v')$ giving 
 $v I v'$. Since this is true for any $u \prec v, u' \prec v'$, we obtain $v I v'$ for all such $v, v'$, 
 contradicting   $\neg(U I U')$. Hence we have $U I U'$.
\end{proof}

\section{Additional material for Section~\ref{sec:main-results}}\label{sec:appendix-main-results}

\xsubparagraph{Counterexample with neighbor antichain boundedness} Consider the
bipartite graph $B_n$ with nodes $\{u_i,v_i\mid i\in\{1,\ldots,n\}\}$, where
$\{u_i,v_j\}$ is an edge if and only if $i\ne j$.  Moreover, let $\cG$ be the
class of graphs containing $B_n$ for every $n\in\N$ and all strongly induced
subgraphs.  Observe that the cliques in $\cG$ have size at most 2: $B_n$ is
bipartite and thus contains cliques of size at most $2$. In particular, if a
clique is a strongly induced subgraph of $B_n$, then it can be of size at most
$2$.  Moreover, the graphs $B_n$ have neighbor antichains of unbounded size:
The set $\{u_1,\ldots,u_n\}$ is a neighbor antichain in $B_n$. 

We currently do not know the exact complexity of $\BSREACH_k(\cG)$. By
\cref{main:k-fixed}, the problem is $\P$-hard and in $\PSPACE$.  Intuitively,
our $\P$ upper bound does not apply because in each node of $\mathcal{R}$, one
would have to remember $n$ bits in order to keep enough information about
commutation of blocks: For a subset $S\subseteq\{1,\ldots,n\}$, let $u_S$ be
the product of all $u_1^+,\ldots,u_n^+$, where we only include $u_i^+$ if $i\in
S$.  Then we have $u_S v_j^-\equiv v_j^+ u_S$ if and only if $j\notin S$.

\section{Additional material for Section~\ref{sec:block-decompositions}}\label{sec:appendix-block-decompositions}
In this section, we provide detailed proofs for \cref{block-decompositions-freely-reducible} and \cref{block-decomposition-few-blocks}.

\subparagraph{Block decompositions and free reductions}
We begin with \cref{block-decompositions-freely-reducible}. For this (and for
\cref{block-decomposition-few-blocks}), we need a simple lemma that allows us
to conclude independencies from relations $R_\pi$.

\begin{lemma}\label{R-vs-I}
	Suppose $\pi$ is a reduction for $w\stackrel{*}{\mapsto_{red}} \varepsilon$.
	If $i,j,k,\ell$ are positions in $w$ such that $i R_\pi j$ and $k R_\pi \ell$
	with $i<k<j<\ell$, then $w[k] I w[j]$.
\end{lemma}
\begin{proof}
	Towards a contradiction, suppose that $w[k]$ and $w[\ell]$ are
	dependent.  Then for every word $w'$ obtained during the reduction
	$\pi$, there are positions $i'<k'<j'<\ell'$ with $w'[i']=w[i]$,
	$w'[k']=w[k]$, $w'[j']=w[j]$, and $w'[\ell']=w[\ell]$. This follows
	directly from the definition of reductions by induction on the number
	of reduction steps: In this situation, we can never move any of the
	letters $w[i],w[k],w[j],w[\ell]$ past each other and hence we can
	never cancel any of them as prescribed in $R_\pi$. Thus, every word
	produced during $\pi$ has at least length $4$, in contradiction to the
	fact that $\pi$ produces $\varepsilon$.
\end{proof}

With \cref{R-vs-I}, we can now prove:
\blockDecompositionsFreelyReducible*
\begin{proof}
	Note that the ``if'' direction is easy to see: If the sequence
	$w_1,\ldots,w_m$ is freely reducible, then there clearly exists a
	reduction $\pi$ so that $w_1,\ldots,w_m$ is a block decomposition with
	respect to $\pi$: We just take the reduction that cancels as the free
	reduction dictates.

	For the converse, suppose there is a reduction $\pi$ of $w_1\cdots w_m$
	such that $w_1,\ldots,w_m$ is a block decomposition with respect to
	$\pi$. Then for each $w_i$, there is a block $w_j$ such that every
	position in $w_i$ cancels either with a letter in $w_i$ or with a
	letter in $w_j$.  In the first step, we consider all those position
	pairs in $w_i$ that are canceled with each other by $\pi$ and cancel
	them using the rule \textbf{F3} to obtain a word $\hat{w}_i$. Now for
	the sequence $\hat{w}_1,\ldots,\hat{w}_m$, we have: (i)~every
	$\hat{w}_i$ is dependent and (ii)~for each $\hat{w}_i$, there is a word
	$\hat{w}_j$ such that every position in $\hat{w}_i$ cancels with a
	position in $\hat{w}_j$ (and vice versa).  A sequence of words for
	which a reduction with properties (i) and (ii) exists will be called
	\emph{block-like}.

	We now prove by induction on $m$ that for every block-like sequence
	$\hat{w}_1,\ldots,\hat{w}_m$, there exists a free reduction (using only
	rules \textbf{F1} and \textbf{F2}). For $m=0$, this is trivial, so
	assume $m>0$.  The property (ii) allows us to define a binary relation
	$\hat{R}$ on $\{1,\ldots,m\}$: We have $i\hat{R}j$ if all the letters
	in $\hat{w}_i$ are $R_\pi$-related to letters in $\hat{w}_j$. We now
	call a pair $(i,j)\in \hat{R}$ with $i<j$ \emph{minimal} if there is no
	other pair $(i',j')\in\hat{R}$ with $i'<j'$ such that $i<i'<j'<j$.

	Clearly, there must be some minimal pair $(i,j)\in\hat{R}$. Now for
	every $\hat{w}_k$ with $i<k<j$, the index $\ell$ with
	$(k,\ell)\in\hat{R}$ must either satisfy $\ell<i$ or $j<\ell$: Any
	other arrangement would contradict minimality of $(i,j)$. By
	\cref{R-vs-I}, this implies that all the letters in $\hat{w}_k$ are
	independent with the letters in $\hat{w}_i$ and $\hat{w}_j$.
	Therefore, we can apply the following sequence of rules:
	\begin{enumerate}
		\item For every $\hat{w}_k$ with $i<k<j$: Move $\hat{w}_k$ past $\hat{w}_j$ using \textbf{FR2}.
		We do this starting with the rightmost such $\hat{w}_k$. We arrive at the sequence 
		\[\hat{w}_1,\ldots,\hat{w}_i,\hat{w}_j,\hat{w}_{i+1},\ldots,\hat{w}_{j-1},\hat{w}_{j+1},\ldots,\hat{w}_m.\]
	\item We can now apply \textbf{FR1} to eliminate $\hat{w}_i,\hat{w}_j$. We thus have the sequence
		\begin{equation}\hat{w}_1,\ldots,\hat{w}_{i+1},\ldots,\hat{w}_{j-1},\hat{w}_{j+1},\ldots,\hat{w}_m,\label{eq-block-to-free-reduction}\end{equation}
		which differs from $\hat{w}_1,\ldots,\hat{w}_m$ only in the removal of $\hat{w}_i$ and $\hat{w}_j$.
	\end{enumerate}
	Now the new sequence \eqref{eq-block-to-free-reduction} is still
	block-like: The reduction $\pi$ that witnesses block-likeness of
	$\hat{w}_1,\ldots,\hat{w}_m$ can clearly be turned into one for
	\eqref{eq-block-to-free-reduction}, because the latter is just missing
	a pair of words that were related by $R_\pi$. Thus, by induction, there
	is a free reduction for \eqref{eq-block-to-free-reduction} and hence
	also for $\hat{w}_1,\ldots,\hat{w}_m$.
\end{proof}

\subparagraph{The number of blocks per context}
We now prove \cref{block-decomposition-few-blocks}: 
\blockDecompositionFewBlocks*
For the proof of \cref{block-decomposition-few-blocks}, given a word $w$ with a decomposition $w=w_1\cdots w_n$ into contexts, we associate with any reduction
$\pi\colon w\stackrel{*}{\mapsto_{red}} \varepsilon$ a refining decomposition. This is the decomposition \emph{induced by $\pi$}, which is defined as follows.

Consider a computation $w \in \mathcal{O}^+$ (that is, $|w|>0$) and a decomposition $w=w_1\cdots w_n$ into contexts. 
Let $\pi\colon w \stackrel{*}{\mapsto_{red}} \varepsilon$ be a reduction that transforms $w$ into $\varepsilon$.

The  \emph{decomposition of $w$ induced by $\pi$} is obtained as follows. 
Parse each context from left to right. 
Suppose we have a context $w_i = w_i[1]w_i[2] \dots w_i[m]$.
Each factor in our new decomposition begins at the leftmost position of a context. 
If the symbol $w_i[k]$ at position $k$ cancels with a symbol $w_j[\ell]$ (that is, $w_i[k] R_{\pi} w_j[\ell]$) in some context $w_j$ 
with $j\ne i$ and no other position in the current factor has canceled with a symbol in context $w_j$, then we start a new factor from position $w_i[k]$. Otherwise, we include $w_i[k]$ in the current factor.

Figure \ref{decomp} illustrates (i)~the decomposition of a computation  over the graph $\overline{\Gamma}_2$ into contexts and (ii)~the decomposition induced by $\pi$. 

Note that this induced decomposition can be defined for any $w\in X_\Gamma^*$
with $w\stackrel{*}{\mapsto_{red}} \varepsilon$. However, in general, it does
not guarantee that the number of factors in a context would be bounded.  For
example, if we have two adjacent vertices $a,b$, then for every $m\ge 0$, the
word $(a^+b^+b^-)^m (a^-)^m$ decomposes into $m+1$ contexts: $m$ times the
context $a^+b^+b^-$ and once the context $(a^-)^m$. However, the last context
decomposes into $m$ factors of the form $a^-$.
\smallskip

Therefore, for \cref{block-decomposition-few-blocks}, we need to establish two
facts, under the assumption that $\pi$ is a reduction witnessing $k$-scope
boundedness of $w$: (i)~the decomposition induced by $\pi$ is a block
decomposition and (ii)~this block decomposition splits each context into at
most $2k$ blocks.  Compared to the block decomposition given in
\cite{DBLP:conf/concur/MeyerMZ18}, the main difference here is that the number
of context switches  is unbounded, whereas \cite{DBLP:conf/concur/MeyerMZ18}
considered computations with bounded number of context switches.

\begin{lemma}
	Let $w=w_1\cdots w_n$ be a context decomposition and let $\pi\colon
	w\stackrel{*}{\mapsto_{red}} \varepsilon$ be a greedy reduction. Then
	the decomposition of $w$ induced by $\pi$ has the following property:
	For any two contexts $w_i$ and $w_j$,    there is at most one factor in
	$w_i$ that is $R_{\pi}$-related to a factor of $w_j$.  \label{lem2}
\end{lemma}
\begin{proof}

	If every position in $w_i$ cancels with a position in $w_i$, then $w_i$
	has only one factor, so this case is trivial. If however, there is a
	position in $w_i$ that cancels with another context, then for each
	factor $f$ in $w_i$ from our decomposition, there is a uniquely
	determined context $w_j$, $j\ne i$, such that every position in $f$
	cancels either with a position in $w_i$ or with a position in $w_j$.
	In this situation, we simply say that $f$ \emph{cancels with} $w_j$.

\begin{figure}[t]
\begin{center}
\includegraphics[scale=0.06]{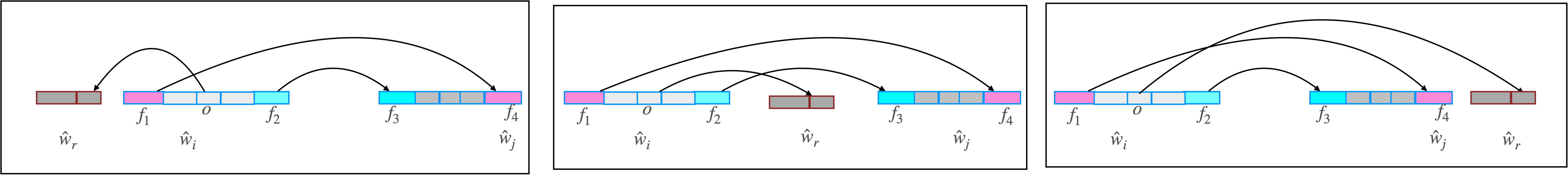}	
\end{center}
\caption{Possible cases for matching in \cref{lem2}}
\label{fig:oneblock}

\end{figure}

	We consider for each context $w_i$ the word $\hat{w}_i$ obtained from
	$w_i$ by removing all positions that cancel with a position in $w_i$
	itself. Since $\pi$ is greedy, we know that $\hat{w}_i$ is irreducible.
	Moreover, the decomposition $\hat{w}_1\cdots\hat{w}_n$ is a
	decomposition into contexts.

	Now $\pi$ yields a reduction $\hat{\pi}$ on the decomposition
	$\hat{w}_1,\ldots,\hat{w}_n$. If we can show our claim for
	$\hat{w}_1,\ldots,\hat{w}_n$ instead of $w_1,\ldots,w_n$, then this
	implies the lemma: If there are two distinct factors in $w_i$ that both
	cancel with $w_j$, then this yields the same situation in $\hat{w}_i$
	and $\hat{w}_j$.

	Toward a contradiction, suppose there were two distinct factors $f_1$
	and $f_2$ in $\hat{w}_i$ that both cancel with $\hat{w}_j$, but there
	is a factor in between them. 

	We begin with some terminology. Consider a position with letter $o$
	between $f_1$ and $f_2$. By assumption, it cancels with some context
	$\hat{w}_r$, $r\ne j$. See \cref{fig:oneblock}. If $i<r<j$, then we
	call this position an \emph{inner position} (because it cancels with a
	context between $\hat{w}_i$ and $\hat{w}_j$). In \cref{fig:oneblock},
	this is the case shown in the middle. If $r<i$ or $j<r$, then we call
	this position an \emph{outer position}. In \cref{fig:oneblock}, these
	are the cases on the left and on the right.

	According to \cref{R-vs-I}, every inner position must commute with
	every letter in $f_2$. Likewise, every outer position must commute with
	every letter in $f_1$. However, commutation between two positions
	inside a context is only possible if both letters belong to the same
	vertex $v$ and one letter is $v^+$ and the other is $v^-$. Therefore,
	the position immediately to the right of $f_1$ must be an inner
	position: Otherwise, $\hat{w}_i$ would not be irreducible.  Likewise,
	the position immediately to the left of $f_2$ must be an outer
	position.  However, that means that somewhere between $f_1$ and $f_2$,
	there must be an inner position directly to the left of an outer
	position. Then, however, \cref{R-vs-I} implies that these two
	neighboring positions must commute, which again means that they can
	cancel directly. This is in contradiction to irreducibility of
	$\hat{w}_i$. 
\end{proof}

\begin{lemma}\label{induced-is-block-decomposition}
	Let $w=w_1\cdots w_n$ be a context decomposition and let $\pi\colon w\stackrel{*}{\mapsto_{red}} \varepsilon$ be a reduction. Then the decomposition induced by $\pi$ is a block decomposition.
\end{lemma}
\begin{proof}
	Suppose $\sc(w)\le k$ and let $\pi\colon w\stackrel{*}{\mapsto_{red}}
	\varepsilon$ be a reduction that witnesses $\sc(w)\le k$.  We first
	show that the decomposition induced by $\pi$ is indeed a block
	decomposition.  Towards a contradiction, suppose there is a factor $f$
	in a context $w_i$ such that there exist distinct factors $g$ and $h$
	in contexts $w_j$ and $w_k$, respectively, such that $i\notin\{j,k\}$
	and (a)~some position in $f$ cancels with a position in $g$ and
	(b)~some position in $f$ cancels with a position in $h$. Then by
	definition of the induced decomposition, we must have $j=k$, because
	otherwise we would have started a new factor mid-way in $f$.  But
	this implies that $w_j$ has two distinct factors ($g$ and $h$) that
	each cancel with some position in $w_i$, which contradicts \cref{lem2}.
\end{proof}

Because of \cref{induced-is-block-decomposition}, from now on, we also speak of the \emph{block decomposition induced by $\pi$}.

We are now ready to prove \cref{block-decomposition-few-blocks}.
\begin{proof}[Proof of \cref{block-decomposition-few-blocks}]
	Suppose $\sc(w)\le k$, let $w=w_1\cdots w_n$ be decomposition into
	contexts, and let $\pi\colon w\stackrel{*}{\mapsto_{red}} \varepsilon$
	be a reduction witnessing $\sc(w)\le k$.  We show that the block
	decomposition induced by $\pi$ splits each context into at most $2k$
	blocks.
	
	Consider a context $w_i$. Let $\{b_1,\ldots,b_r\}$ the set of blocks
	in $w_i$. If each position in $w_i$ cancels with a position in $w_i$,
	then $w_i$ consists of a single block.  Thus, we assume that at least
	one position in $w_i$ cancels with a position in some other context.
	In that case, for each block $b_\ell$ in $w_i$, there is a unique
	context $w_{f(\ell)}$ such that every position in $b_\ell$ cancels
	with some position in $w_i$ or in $w_{f(\ell)}$. By \cref{lem2}, we
	know that the function $f\colon\{1,\ldots,\ell\}\to\{1,\ldots,n\}$ is
	injective.  Moreover, since $\pi$ is a witness for $k$-scope
	boundedness, we know that each context $w_{f(\ell)}$ has interaction
	distance at most $k$ from $w_i$. Moreover, $w_{f(\ell)}$ must belong
	to the same weak dependence class as $w_i$. Since there are at most
	$2k$ blocks at interaction distance $k$ in the same weak dependence
	class as $w_i$, we know that the image $\{f(1),\ldots,f(r)\}$ has at
	most $2k$ elements. Since $f$ is injective, this implies $r\le 2k$.
\end{proof}

\section{Additional material for Section~\ref{sec:decision-procedure}}\label{sec:appendix-decision-procedure}

\begin{algorithm}
	\caption{$\mathsf{BSREACH}(\Gamma,k)$ : $k$-scope bound reachability for a Valence System}
	\begin{algorithmic}[1]
		\Require  
		Valence System $\cA$ over a graph monoid  $\mathbb{M}\Gamma$, 
		scope bound $k$.  
		
		\Ensure ACCEPT if there a $k$ scope bounded computation $w \in X^{*}_{\Gamma}$, else REJECT. 
		
		\State $q \gets q_{init}$
		
		\State $o \gets \phi$

		\State $class \gets \phi$
		
		\For{$x \in \{1, \cdots, d\}$, where $d = |[~]_{\sim_W}|$}		
			\State $s_{\gamma_x} \gets s_{init} = \left((\mathcal{C}_1, \varepsilon, \varepsilon), (\mathcal{C}_2, \varepsilon,\varepsilon),\cdots, (\mathcal{C}_k, \varepsilon,\varepsilon)\right)$, where $\mathcal{C}_c = (\mathcal{E},  \mathcal{E}, \cdots,\mathcal{E})$
		\EndFor

		\State \textcolor{red}{/* Initialize the configuration  */} 
		
		\While{True}
		
			\State Guess an equivalence class $x \in \{1, \cdots, d\}$, where $d = |[~]_{\sim_W}|$
			
			\State $\forall i \in \{1, \cdots, 2k\}$ $q^c_i \gets$ Guess end state for each block abstraction $i$
			
			\State $\forall i \in \{1, \cdots, 2k\}$ Guess $\mini(B_i), \maxi(B_i)$ such that
			
			\State $~~~~~\forall i, j$, $\maxi(B_i)$ and $\maxi(B_j)$ are dependent

			\textcolor{red}{/* Symbols of blocks $B_1, \dots, B_{2k}$ mutually dependent */}

			\State $\forall i \in \{1, \cdots, 2k\}$ Guess $f^c_i \in V$, the start symbol for each block
			
			\If{$\neg (f^c_1 ~ I ~ o)$} \textbf{REJECT}
			\EndIf
			
			\State $\forall i \in \{1, \cdots, 2k\}$ Guess $o^c_i \in V$, a symbol that necessarily appears in the block 
			
			\State Guess $o_c \in \{o^c_1, \cdots, o^c_{2k}\}$ which appears in context $\mathcal{C}_c$ 

			\textcolor{red}{/* $o_c$ is used to check independence with the first symbol of next context */}

			\State Thus, form the next context $\mathcal{C}_c = (N^c_1, \cdots, N^c_{2k}, f_c, o_c)$ for $\gamma_x$
			
			where $q^c_0 \gets q$ \textcolor{red}{/* start from the last state of previous context */}

			$f_c = f^c_1$
			
			and $N^c_i = (q^c_{i-1}, q^c_i, f^c_i, o^c_i, \mini(B_i), \maxi(B_i))$ $\forall i \in \{1, \cdots, 2k\}$

			\State Guess a sequence of free reduction operations $\rightsquigarrow_{red}$ (possibly 0 operations), \State to eliminate the leftmost context

			\textcolor{red}{/* This converts $ N^1_1, \cdots, N^1_{2k}, N^2_1, \cdots, N^2_{2k}, \cdots N^k_1, \cdots, N^k_{2k}, N^c_1, \cdots, N^c_{2k}$} 

			\textcolor{red}{	to $ N'^1_1, \cdots, N'^1_{2k}, N'^2_1, \cdots, N'^2_{2k}, \cdots N'^k_1, \cdots, N'^k_{2k}, N'^c_1, \cdots, N'^c_{2k}$ */}

			\If{$\exists i \in \{1, \cdots, 2k\} ~ N'^1_i \neq \mathcal{E}$}
				\textbf{REJECT}
			\EndIf

			\State $q \gets q^c_{2k}$ \textcolor{red}{/* Store last state of the context in $q$ */}
			\State $o \gets o_c$ \textcolor{red}{/* Store some seen symbol from this context in $o$ */}
			\State $class \gets x$
			
			\State $s_{\gamma_x} \gets ((\mathcal{C}'_2, f_2, o_2) \cdots, (\mathcal{C}'_k, f_k,o_k), (\mathcal{C}'_c, f_c, o_c))$ 
			
			\If{$\forall \gamma \in [~]_{\sim_W} \ \forall c \in \{1, \cdots, k\} \ \forall i \in \{1, \cdots, 2k\} \ N^c_i = \mathcal{E}$ and $q = q_{fin}$}
				\State \textbf{ACCEPT}
			\EndIf

		\EndWhile
		
	\end{algorithmic}
	\label{algo1}
	\end{algorithm}

\subsection{Proof of Lemma~\ref{check-blocks-cancellable}}
\label{app:check-blocks-cancellable}
Given a context abstraction $\mathcal{C}=(N_{bl}^1, \dots, N_{bl}^{2k}, f, o)$ where each $N_{bl}^i$ has the form \\
  $(q_0^i, q_1^i, f_{bl},o_{bl}, U_i^\mini,U_i^\maxi)$, one can check whether for any pair of block abstractions
  $N_{bl}^i, N_{bl}^j$, they are over mutually dependent symbols. To do this, it suffices to check if  
  $U_i^\maxi$ is dependent with $U_j^\maxi$. That is, each 
  $u \in  U_i^\maxi$ is dependent with each $u'\in U_j^\maxi$. 
  
  \begin{lemma}(Dependence check between block abstractions) 
  Two block abstractions $N_{bl}^i, N_{bl}^j$ are over mutually dependent symbols 
  iff $U_i^\maxi$ is dependent with $U_j^\maxi$.
  \label{lem:dep}
    \end{lemma}
   \begin{proof}
   	Assume that $N_{bl}^i, N_{bl}^j$ are over mutually dependent symbols. 
  Then clearly,  $U_i^\maxi$ is dependent with $U_j^\maxi$. 
  Conversely, assume that $U_i^\maxi$ is dependent with $U_j^\maxi$. 
  Consider $v_i \in U_i^\maxi$, $w_i \in B_i$ such that 
  $w_i <v_i$, and likewise $v_j \in U_j^\maxi$ and $w_j \in B_j$, with $w_j < v_j$. 
  Assume $w_i I w_j$. Then $w_i \in N(v_j), w_j \in N(v_i)$ resulting in 
  $v_i \in N(v_j), v_j \in N(v_i)$. 
  This  contradicts  the mutual dependence of $U_i^\maxi$ and 
  $U_j^\maxi$. 
   \end{proof}

 \subsection*{Checking a free reduction between block abstractions} 
      
\noindent\textbf{Check if two blocks commute}.    With the representation of block abstractions, to check if two block abstractions  $N_u, N_v$ commute, it is enough to check 
     if $U_u^\mini I U_v^\mini$ (Lemma \ref{mini}). 
     \smallskip

     \noindent\textbf{Check if two blocks cancel}. We now discuss how to use block abstractions to check if two blocks $N_u, N_v$ cancel with each other. Given the block abstractions $N_u=(q_1, q_2, f_u, o_u, U_u^\mini, U_u^\maxi)$ and $N_v$=$(q_3, q_4, f_v, o_v, U_v^\mini, U_v^\maxi)$, we first compute 
\begin{align*}
     B_i& =\{x\in V \mid \exists y\in U_i^\maxi\colon x\preceq y\}  \\
     \hat{B}_i& =\{x\in B_i \mid \exists y\in U_i^\mini\colon y\preceq x\}
\end{align*}
for $i\in\{u,v\}$.
and the corresponding alphabets
\begin{align*}
	X_i& =\{o^+,o^-\mid o\in B_i\} \\
	\hat{X}_i& =\{o^+,o^- \mid o\in\hat{B}_i\}
\end{align*}
for $i\in\{u,v\}$.

Then, by definition of cancellation, we have to check whether there are words $w_u,\hat{w}_u,w_v,\hat{w}_v$, such that
\begin{enumerate}
	\item\label{cancellation-a} $w_u$ is read one some path from $q_1$ to $q_2$
	\item \label{cancellation-b}$w_v$ is read one some path from $q_3$ to $q_4$
	\item \label{cancellation-c}$w_i\in X_i^*$, $\hat{w}_i\in \hat{X}_i^*$ for $i\in\{u,v\}$,
	\item \label{cancellation-d}$w_i \stackrel{*}{\mapsto_{red}} \hat{w}_i$ using rules \textbf{R1} and \textbf{R2}, for $i\in\{u,v\}$
	\item \label{cancellation-e}$\hat{w}_u\hat{w}_v\equiv_\Gamma\varepsilon$.
	\item \label{cancellation-f}$o_i$ occurs in $w_i$, for $i\in\{u,v\}$
	\item \label{cancellation-g}$w_i$ begins with $f_i$, for $i\in\{u,v\}$
\end{enumerate}
Note that since $\hat{B}_i\subseteq B_i$, if such words exist, then they
exist such that $\hat{w}_u$ and $\hat{w}_v$ are irreducible. In that case, we
must have $\hat{w}_u,\hat{w}_v\in (\hat{X}_u\cap \hat{X}_v)^*$.

We now construct a PDA $\mathcal{P}_{uv}$ for the language of all words
$w_u\#w_v$ such that words $\hat{w}_u$ and $\hat{w}_v$ exist so that
conditions
\labelcref{cancellation-a,cancellation-b,cancellation-c,cancellation-d,cancellation-e}
are satisfied. This proves that we can decide cancellation in polynomial time:
By intersecting with the regular language
\[ \{ w_u\#w_v \mid \text{$o_i$ occurs in $w_i$, and $w_i$ begins with $f_i$, for $i\in\{u,v\}$}\}, \]
we obtain a PDA that accepts a non-empty set if and only if the two block abstractions cancel.

The PDA $\mathcal{P}_{uv}$ is constructed as follows. Intuitively, we first
read the word $w_u$ and on the stack, we leave the word $\hat{w}_v$. To this end,
the stack alphabet contains $\hat{X}_u\cap\hat{X}_u$ and also separate symbols
$\{(x,i)\mid x\in X_i\setminus (\hat{X}_u\cap \hat{X}_v)\}$ for $i\in\{u,v\}$. Here, the letters
$(x,i)$ are used to simulate internal cancellations in the word $w_i$, $i\in\{u,v\}$.
Note that we have to make sure that $\hat{w}_i$ does not contain any letters from $X_i\setminus \hat{X}_i$. Therefore, while reading $w_v$,
our PDA cannot pop letters $(x,u)$, so that it is guaranteed that 
the word $\hat{w}_u$ contains only letters in $\hat{X}_u$.
Then, after
reading $\#$, the automaton reads $w_v$ and and makes sure that $\hat{w}_v$ cancels with the stack content that existed when it read $\#$.

Its stack alphabet is $\Delta=(\hat{X}_u\cap \hat{X}_v)\cup \{(x,i) \mid x\in X_i\setminus(\hat{X}_u\cap \hat{X}_v),~i\in\{u,v\}\}$.
Its set of states is $Q\times\{u,v\}$, where the states $(q,u)$ are used to read $w_u$ and the states $(q,v)$ are used to read $w_v$.
We have the following transitions:
\begin{itemize}
	\item For every $x\in \hat{B}_u\cap\hat{B}_v$ and $i\in\{u,v\}$:
		\begin{itemize}
			\item For every transition $p\xrightarrow{x^+}q$, we add a transition from $(p,i)$ to $(q,i)$ that pushes $x^+$ and reads $x^+$.
			\item For every transition $p\xrightarrow{x^-}q$, we add a transition from $(p,i)$ to $(q,i)$ that pops $x^+$ and reads $x^-$.
			\item If $xIx$, then for every transition $p\xrightarrow{x^-}q$, we add a transition from $(p,i)$ to $(q,i)$ that pushes $x^-$ and reads $x^-$.
			\item If $xIx$, then for every transition $p\xrightarrow{x^+}q$, we add a transition from $(p,i)$ to $(q,i)$ that pops $x^-$ and reads $x^+$.
		\end{itemize}
	\item For every $x\in B_i\setminus (\hat{B}_u\cap \hat{B}_v)$ and $i\in\{u,v\}$:
		\begin{itemize}
			\item For every transition $p\xrightarrow{x^+}q$, we add a transition from $(p,i)$ to $(q,i)$ that pushes $(x^+,i)$ and reads $x^+$.
			\item For every transition $p\xrightarrow{x^-}q$, we add a transition from $(p,i)$ to $(q,i)$ that pops $(x^+,i)$ and reads $x^-$.
			\item If $xIx$, then for every transition $p\xrightarrow{x^-}q$, we add a transition from $(p,i)$ to $(q,i)$ that pushes $(x^-,i)$ and reads $x^-$.
			\item If $xIx$, then for every transition $p\xrightarrow{x^+}q$, we add a transition from $p$ to $q$ that pops $(x^-,i)$ and reads $x^+$.
		\end{itemize}
	\item We add a transition from $(q_2,u)$ to $(q_3,v)$, reading $\#$.
\end{itemize}
Finally, the initial state of our PDA is $(q_1,u)$ and its final state is $(q_4,v)$.

\subsection{Proof of Proposition~\ref{upper-bound-k-input-simple}}
We now come to the proof of \cref{upper-bound-k-input-simple}:
\upperBoundKInputSimple*
We prove each of the two statements in a separate lemma.
\begin{restatable}{lemma}{upperBoundKInputAntiClique}\label{upper-bound-k-input-anti-clique}
  If $\cG$ only contais anti-cliques, then $\BSREACH(\cG)$ belongs
  to $\P$. 
\end{restatable}
\begin{proof}
Suppose we are given a graph $\Gamma=(V,I)$ in $\cG$, a valence system $\cA=(Q,\to)$
over $\Gamma$, states $q_0,q_f\in Q$, and a scope bound $k$.  

Observe since no two distinct vertices are adjacent in $\Gamma$, every word
over $X_\Gamma$ contains only one context. Therefore, $\BSREACH$ is just the
problem whether there is any word $w\in X_\Gamma^*$ with $w\equiv_\Gamma 1$ and
$(q_0,\varepsilon)\autsteps(q_f,w)$. We reduce this problem to emptiness of
pushdown automata.

We construct a pushdown automaton $\calP$ as follows. It has the state set $Q$,
initial state $q_0$, and final state $q_f$. It stack alphabet is $\Delta=\{v^+,v^-,v^{\circ}\mid v\in V\}$. The transitions are as follows:
\begin{enumerate}
\item For every transition $p\xrightarrow{v^+} q$, add a transition from $p$ to $q$ that pushes $v^+$ on the stack.
\item For every transition $p\xrightarrow{v^-} q$, add a transition from $p$ to $q$ that pops $v^+$ from the stack.
\item For every transition $p\xrightarrow{v^-} q$ such that $v$ has a self-loop in $\Gamma$, add a transition from $p$ to $q$ that pushes $v^\circ$ on the stack.
\item For every transition $p\xrightarrow{v^+} q$ such that $v$ has a self-loop in $\Gamma$, add a transition from $p$ to $q$ that pops $v^\circ$ from the stack.
\end{enumerate}
Each of these transitions reads the empty word $\varepsilon$ from the input.

It is now straightforward to check that $(q_0,\varepsilon)\autsteps(q_f,w)$ for
some $w\in X_\Gamma^*$ with $w\equiv 1$ if and only if $\calP$ accepts the empty
word, if acceptance is defined by reaching the configuration
$(q_f,\varepsilon)$.  Since membership of pushdown automata can be decided in
polynomial time (see, e.g.~\cite[Proposition 5.8 and Theorem 4.30]{Sippu1988},
the \lcnamecref{upper-bound-k-input-anti-clique} follows.  
\end{proof}

\begin{restatable}{lemma}{upperBoundKInputSingleton}\label{upper-bound-k-input-singleton}
  If $\Gamma$ has a single vertex, then $\BSREACH{\Gamma}$
  belongs to $\NL$.
\end{restatable}
\begin{proof}
Since $\Gamma=(V,I)$ has a single vertex, we have $X_\Gamma=\{v^+,v^-\}$
for the vertex $v\in V$.  Furthermore, let $\cA$ be a valence system over $\Gamma$.

Since every subset of $X_\Gamma$ is a single context,
$\BSREACH$ is just the problem of whether there exists a word $w\in X_\Gamma^*$
with $(q_0,\varepsilon)\autsteps(q_f,w)$ such that $w\equiv_\Gamma \varepsilon$.
The condition $w\equiv_\Gamma \varepsilon$ is a simple counting condition:
\begin{enumerate}
\item If $v$ has a self-loop, then $w\equiv_\Gamma \varepsilon$ if and only if
$|w|_{v^+}=|w|_{v^-}$.
\item If $v$ has no self-loop, then $w\equiv_\Gamma \varepsilon$ if and only if
$|w|_{v^+}=|w|_{v^-}$ and for every prefix $u$ of $w$, we have $|w|_{v^+}\ge |w|_{v^-}$.
\end{enumerate}
In each case, we shall see that it is easy to construct a one-counter automaton
for the set of $w\in X_\Gamma^*$ such that $(q_0,\varepsilon)\autsteps(q_f,w)$
and $w\equiv_\Gamma \varepsilon$. Since non-emptiness of one-counter automata
is decidable in $\NL$ (see, e.g.~\cite{DBLP:conf/time/DemriG07}), this implies the
\lcnamecref{upper-bound-k-input-singleton}.

First, suppose $v$ has a self-loop. Then, after reading a prefix $u$, the
one-counter automaton stores the number $|u|_{v^+}-|u|_{v^-}$ using its
counter. Here, if the number is negative, the counter contains its absolute
value and remembers the sign in its state.

If $v$ has no self-loop, then after reading a prefix $u\in X_\Gamma^*$, then
the one-counter automaton also stores $|u|_{v^+}-|u|_{v^-}$ on the
counter. However, now it must make sure that this value never drops below zero,
meaning that it does not need the extra bit for the sign.
\end{proof}

\section{Additional material for Section~\ref{sec:hardness}}\label{sec:appendix-hardness}
\subsection{Proof of Proposition~\ref{lower-bound-k-fixed-two-nonadjacent}}
We prove the following:
\lowerBoundKFixedTwoNonadjacent*
\begin{proof}
	Suppose $\Gamma=(V,I)$ and $u,v\in V$ are non-adjacent.  A valence
	system that only uses the vertices corresponding to $u,v$ will only
	have runs that consist of one context. Therefore, the general 
	reachability problem reduces to $\BSREACH_k(\Gamma)$ for each $k\ge 1$.

	If both $u$ and $v$ are looped, the general reachability problem is
	the rational subset membership for the graph group of the graph
	consisting of two non-adjacent vertices, for which $\P$-hardness was
	observed in \cite[Theorem III.4]{DBLP:conf/lics/HaaseZ19}. (More
	precisely, our problem here is membership of the neutral element in a
	rational subset, but the general rational subset membership problem
	reduces to this case in logspace~\cite{Lohrey2013a}).

	If both $u$ and $v$ are unlooped, then it is easy to reduce emptiness
	of a pushdown automaton to general reachability: We may assume that
	the pushdown automaton has a binary stack alphabet, $\{u, v\}$. We
	turn the pushdown automaton into a valence system over
	$\Gamma$ as follows: If it pushes $x$ for
	$x\in\{u,v\}$, we multiply $x^+$. If it pops $x$, we multiply
	$x^-$. Then clearly, the pushdown automaton has a run from the
	empty stack to the empty stack if and only if the reachability instance
	is positive.

	Finally, suppose $u$ is unlooped and $v$ is looped. Let $\Delta$ be the
	graph with vertices $a$ and $b$, where both $a$ and $b$ are unlooped. By
	our previous argument, we know that $\BSREACH_k(\Delta)$ is $\P$-hard.
	We reduce $\BSREACH_k(\Delta)$ to $\BSREACH_k(\Gamma)$ as follows.
	Given a valence system $\cA$ over $\Delta$, we construct a valence system
	$\hat{\cA}$ over $\Gamma$ with the same set of states and the following transitions:
	\begin{align*}
		&p\xrightarrow{u^+v^+} q   && \text{for every transition $p\xrightarrow{a^+}q$} \\
		&p\xrightarrow{v^-u^-} q && \text{for every transition $p\xrightarrow{a^-} q$}. \\
		&p\xrightarrow{u^+v^+v^+} q   && \text{for every transition $p\xrightarrow{b^+}q$} \\
		&p\xrightarrow{v^-v^-u^-} q && \text{for every transition $p\xrightarrow{b^-} q$}.
	\end{align*}

	Moreover, $\hat{\cA}$ has the same initial and final state as $\cA$.
	Then it is easy to see that $\hat{\cA}$ is a positive instance if and
	only if $\cA$ is a positive instance.
\end{proof}

\subsection{Construction in proof of Proposition~\ref{lower-bound-k-input-two-adjacent}}\label{lower-bound-k-input-two-adjacent-details}
   Let $\Gamma=(V,I)$ and suppose $u,v\in V$ are adjacent.
  Let $\cA=(Q,T,q_0,w_0)$ be a BQA and $q\in Q$ be a state.

  As above, we use the shorthands
  \[ \bitzero=u^+,~~~\bar{\bitzero}=u^-,~~~\bitone=u^+u^+,~~~\bar{\bitone}=u^-u^-,~~~\n=v^+v^-,~~~\biteps=u^+u^-.\]

  We construct a valence automaton
  $\hat{\cA}=(\hat{Q},\hat{T},\hat{q}_0,\hat{q}_f)$ over $\Gamma$ as
  follows. Its state set is $\hat{Q}=Q\cup\{\hat{q}_0,\hat{q}_f\}$,
  where $\hat{q}_0$ and $\hat{q}_f$ are fresh (initial and final)
  states.

  In the beginning, our valence automaton multiplies
  \[ t_0=(\bitzero\n\bitone\n\bitone\n)(\biteps\n\biteps\n\biteps\n\bitzero\n\bitone\n\bitone\n)^{n-1},\]
  meaning there is a transition $(\hat{q}_0,t_0,q_0)$.
  Moreover, in the end, our valence system multiplies
  \[ t_f=(\bar{\bitzero}\n\bar{\bitone}\n\bar{\bitone}\n)(\biteps\n\biteps\n\biteps\n\bar{\bitzero}\n\bar{\bitone}\n\bar{\bitone}\n)^{n-1}. \]
  This means there is a transition $(q_f,t_f,\hat{q}_f)$.
  Finally, for each transition $(p,x,y,p')$ in $\cA$, we have a transition
  $(p,t_{x,y},p')$ in $\hat{\cA}$, where
  \[ t_{x,y}=x_1\n x_2 \n x_3 \n y_1\n y_2\n y_3\n, \]
  in which
  \[ (x_1,x_2,x_3)=\begin{cases} (\bar{\bitzero},\bar{\bitone},\bar{\bitone}) & \text{if $x=0$}, \\
      (\bar{\bitone}, \bar{\bitzero},\bar{\bitzero}) & \text{if $x=1$}, \end{cases}~~~\text{and}~~~(y_1,y_2,y_3) = \begin{cases} (\bitzero, \bitone,\bitone) & \text{if $y=0$}, \\ (\bitone,\bitzero,\bitzero) & \text{if $y=1$}\end{cases}. \]

  Finally, if the vertices $1$ and $2$ in $\Gamma$ belong to different
  weak dependence classes, we set $k=3(2n-1)$. If they are weakly
  dependent, we set $k=6(2n-1)$. As described in \cref{sec:hardness}, it is
  now straightforward to verify that there is a $k$-scope-bounded run
  from $\hat{q}_0$ to $\hat{q}_f$ in $\hat{\cA}$ if and only if
  $(q_0,0^n)\autsteps(q_f,0^n)$ in $\cA$.

\subsection{Construction in proof of Proposition~\ref{lower-bound-k-fixed-unbounded-cliques}}

  We  construct a valence system $\hat{\cA}$ as follows. It has states $\hat{Q}=Q\cup\{\hat{q}_0, \hat{q}_f\}$,
  where $\hat{q}_0$ and $\hat{q}_f$ are its final and initial state, respectively. We use the idea sketched to maintain the $i$-th bit using $a_i$, $\bar{a}_i$, $b_i$, and $\bar{b}_i$.
  Thus, we first initialize all $n$ bits by using the string
  \[ t_{0} =t_{0,1}\cdots t_{0,n},~~~\text{where}~~~t_{0,i}= a_i^{k}b_i\bar{b}_i~\text{for $i=1,\ldots,n$}. \]
  This means, there is an edge  $(\hat{q}_0,t_0,q_0)$. At the end of the computation,
  we make sure that all bits are set to zero using the string
  \[ t_{f} = t_{f,1}\cdots t_{f,n},~~~\text{where}~~~t_{f,i}=\bar{a}_i^{k}b_i\bar{b}_i~\text{for $i=1,\ldots,n$}. \]
  Thus, there is an edge $(q_f,t_f,\hat{q}_f)$.

  In order to simulate a transition $(p,i,x,y,p')$ of the BVA, we have to first read the value $x$ from counter $i$, using $\bar{a}_i^{(1+2x)k}b_i\bar{b}_i$, then use $(a_i\bar{a}_ib_i\bar{b}_i)^k$ to create a large interaction distance, and finally use $a_i^{(1+2y)k}b_i\bar{b}_i$ to store the value $y$. Hence, for every transition $(p,i,x,y,p')$ in the BVA, we have a transition
 $(p,t_{i,x,y},p')$ in  $\hat{\cA}$, where
 \[ t_{i,x,y} = \bar{a}_i^{(1+2x)k} b_i\bar{b}_i (a_i\bar{a}_ib_i\bar{b}_i)^k a_i^{(1+2y)k} b_i\bar{b}_i.\]
 Note that we now indeed use $a_i^kb_i\bar{b}_i$ to store $0$ and $a_i^{3k}b_i\bar{b}_i$ to store $1$ (and analogously for reading the bits).
 Hence, by the observation, we have $(q_0,0^n)\autsteps (q_f,0^n)$ in $\cA$ if and only if $(\hat{q}_0,\varepsilon)\autsteps (\hat{q}_f,w)$ for some $w\in X_\Gamma^*$ with $w\equiv 1$ and $\mathsf{sc}(w)\le k$.

 \subsection{Proofs of Theorem~\ref{main:k-input} and Theorem~\ref{main:k-fixed}}\label{sec:appendix-hardness-main-results}

\subparagraph{Proof of \cref{main:k-input}} Let $\cG$ be a graph class. If all
the graphs in $\cG$ are anti-cliques, then the upper bounds are provided by
\cref{upper-bound-k-input-simple}.  $\NL$-hardness is inherited from the
reachability problem in finite directed graphs. Moreover, $\P$-hardness in the
case that $\cG$ contains a graph with two non-adjacent vertices follows from
\cref{lower-bound-k-fixed-two-nonadjacent}.  If there is a graph $\Gamma$ in
$\cG$ that is not an anti-clique, then $\Gamma$ has two non-adjacent vertices
and thus \cref{lower-bound-k-input-two-adjacent} yields $\PSPACE$-hardness of
$\BSREACH(\Gamma)$ and hence of $\BSREACH(\cG)$.

\subparagraph{Proof of \cref{main:k-fixed}} Now let $k\ge 1$ be fixed and let
$\cG$ be a class of graphs that is closed under strongly induced subgraphs. If
$\cG$ only contains cliques and the graphs in $\cG$ have bounded size, then
$\BSREACH_k(\cG)$ is in $\NL$ by
\cref{upper-bound-k-fixed-cliques-bounded-cliques}. Moreover, if the cliques in
$\cG$ have bounded size, then $\BSREACH_k(\cG)$ is in $\P$ by
\cref{upper-bound-k-fixed-bounded-cliques}.  As above, $\NL$-hardness is
inherited from reachability in finite directed graphs.  Furthermore,
$\P$-hardness follows from \cref{lower-bound-k-fixed-two-nonadjacent}.
Finally, if these two conditions are not met, then $\cG$ contains arbitrarily
large cliques. Since $\cG$ is closed under strongly induced subgraphs, this
implies that $\cG$ either (i)~contains an unlooped clique of every size or
(ii)~contains a looped clique of every size. Thus $\PSPACE$-hardness of
$\BSREACH_k(\cG)$ follows from \cref{lower-bound-k-fixed-unbounded-cliques}.

\end{document}